\documentclass[a4paper]{scrarticle}
\usepackage{mathtools}
\usepackage{amsthm}
\usepackage{amssymb}
\usepackage{amsmath}
\usepackage{dsfont}
\usepackage{url}
\usepackage{hyperref}
\usepackage[capitalise]{cleveref}
\usepackage{algorithm2e}
\usepackage[utf8]{inputenc}
\usepackage[T1]{fontenc}
\usepackage{tikz}
\usepackage{mathrsfs} % for mathscr
\usepackage{centernot}
\usepackage{mdframed}
\usetikzlibrary{automata,positioning}
\usepackage{totcount}

\newtotcounter{todoctr}

\newcommand{\tikzxshift}[1]{\tikzset{xshift=#1}}
\newcommand{\tikzyshift}[1]{\tikzset{yshift=#1}}

\newcommand{\startdate}[1]{}

\newcommand{\head}[1]{\textbf{#1.}}

\newcommand{\iitv}[1]{[\![#1]\!]}

\newcommand{\es}{\varnothing}
\newcommand{\subs}{\subseteq}

\theoremstyle{plain}
\newtheorem{thm}{\textsf{Theorem}}
\crefname{thm}{theorem}{theorems}
\newtheorem{prop}[thm]{\textsf{Proposition}}
\crefname{prop}{proposition}{propositions}

\crefname{prop}{property}{properties}
\newtheorem{cor}[thm]{\textsf{Corollary}}
\crefname{cor}{corollary}{corollaries}
\newtheorem{lem}[thm]{\textsf{Lemma}}
\crefname{lem}{lemma}{lemmas}
\Crefname{lem}{Lemma}{Lemmas}

\crefname{pblm}{problem}{problems}

\theoremstyle{remark}
\newtheorem{expl}[thm]{\textsf{Example}}

\crefname{req}{remark}{remarks}

\theoremstyle{definition}
\newtheorem{defn}[thm]{\textsf{Definition}}

\usepackage{bm}
\usepackage{algorithm2e}
\usepackage{mathdots}
\usetikzlibrary{graphs}
\usetikzlibrary{graphs.standard}

\usepackage{tcolorbox}

\usepackage{multicol}
\usepackage{caption}
\usepackage{subcaption}
\usepackage{forest}
\usepackage[margin=50pt,footskip=1cm]{geometry}

\newcommand{\supp}{\operatorname{supp}}
\DeclareMathOperator*{\argmin}{arg\,min}

\newcommand{\Mm}{\mathcal{M}}

\newcommand{\N}{\mathbb{N}}
\newcommand{\DD}{\mathbb{D}}
\newcommand{\FF}{\mathbb{F}}
\newcommand{\Dd}{\mathcal{D}}
\newcommand{\Aa}{\mathcal{A}}

\newcommand{\Tt}{\mathcal{T}}

\newcommand{\Ff}{\mathcal{F}}
\newcommand{\Cc}{\mathcal{C}}
\newcommand{\Gg}{\mathcal{G}}

\newcommand{\troot}{\operatorname{root}}
\newcommand{\npreds}{\operatorname{npreds}}

\newcommand{\depth}{\operatorname{depth}}

\newcommand{\bott}[1]{\textbf{#1}}
\newcommand{\tildeop}[1]{\widetilde{#1}}
\newcommand{\meet}{\wedge}

\newcommand{\bigmeet}{\bigwedge}
\newcommand{\bigjoin}{\bigvee}

\providecommand{\keywords}[1]
{
  \textbf{\textit{Keywords---}} #1
}

\title{Factorisation in the semiring of finite dynamical
  systems}
\date{\today}
\author{\'Emile Naquin\thanks{\'Ecole Normale Sup\'erieure de Lyon, France. Email: \texttt{emile.touileb@ens-lyon.fr}}, 
Maximilien Gadouleau\thanks{Department of Computer Science, Durham University, UK. Email: \texttt{m.r.gadouleau@durham.ac.uk}}
}

\begin{document}
\maketitle
%\maketodo

\begin{abstract}
Finite dynamical systems (FDSs) are commonly used to model systems with a finite number of states that evolve deterministically and at discrete time steps. Considered up to isomorphism, those correspond to functional graphs. As such, FDSs have a sum and product operation, which correspond to the direct sum and direct product of their respective graphs; the collection of FDSs endowed with these operations then forms a semiring. The algebraic structure of the product of FDSs is particularly interesting. For instance, an FDS can be factorised if and only if it is composed of two sub-systems running in parallel. In this work, we further the understanding of the factorisation, division, and root finding problems for FDSs. Firstly, an FDS $A$ is cancellative if one can divide by it unambiguously, i.e. $AX = AY$ implies $X = Y$. We prove that an FDS $A$ is cancellative if and only if it has a fixpoint. Secondly, we prove that if an FDS $A$ has a $k$-th root (i.e. $B$ such that $B^k = A$), then it is unique. Thirdly, unlike integers, the monoid of FDS product does not have unique factorisation into irreducibles. We instead exhibit a large class of monoids of FDSs with unique factorisation. To obtain our main results, we introduce the unrolling of an FDS, which can be viewed as a space-time expansion of the system. This allows us to work with (possibly infinite) trees, where the product is easier to handle than its counterpart for FDSs.
\end{abstract}

\keywords{finite dynamical systems; factorisation; cancellative elements; trees; graph direct product}

\section{Introduction}

Finite dynamical systems are commonly used to model systems with a finite number of states that
evolve deterministically and at discrete time steps. Multiple models have been proposed for various settings, such as Boolean networks \cite{BoolNets, Goles90}, reaction systems \cite{ReactionSystems}, or sandpile models \cite{BTW87}, with applications to biology \cite{Tho73, TD90, BCRCGD13}, chemistry \cite{ReactionSystems}, or information theory \cite{GR11,GRR15}.

The dynamics of an FDS are easily described via its graph, which consists of a collection of cycles containing the periodic states, to which are attached tree-like structures containing the transient states. As such, two families of FDSs are of particular interest: permutations only have disjoint cycles in their graphs, while the so-called dendrons, where all states eventually converge towards the same fixed point, only have a tree in their graphs. Therefore, any FDS can be viewed as a collection of dendrons attached to a given permutation.

Given two FDSs $A$ and $B$, we can either \textit{add} them (that is, create a system
that behave like $A$ when it starts in a state of $A$, and like $B$
when it starts in a state of $B$) or \textit{multiply} them (that is,
create a system that corresponds to $A$ and $B$ evolving in parallel). Thus, the set $\DD$ of FDSs, endowed with the sum and product above, forms a semiring.

Since the introduction of the semiring of finite
dynamical systems (FDSs) in \cite{PolEqFDS} as an abstract way of
studying FDSs, some research has been devoted to understand more
thoroughly the multiplicative structure of this semiring \cite{Dor17, ComposingBehaviors, Couturier, GE20}. 
We can highlight three important problems related to the multiplicative structure of $\DD$.
\begin{enumerate}
    \item Perhaps the most obvious problem is factorisation: given an FDS $C$,
can we find two non-trivial FDSs $A$ and $B$ (with fewer states than $C$) such that $C=A\times B$? This corresponds to whether the system modelled by $C$ is actually composed of two independent parts working in parallel.

In \cite{Dor17, ComposingBehaviors}, it is shown that the answer is usually negative: the proportion of reducible FDSs of size $n$ vanishes when $n\rightarrow\infty$. Moreover, unlike for integers, the semiring $\DD$ does not have unique factorisation into irreducible elements. Worse yet, this is true when we restrict ourselves to permutations or to dendrons. This adds another layer of difficulty for problems related to factorisation in the semiring of FDSs.

    \item Another important problem is division: given $C$ and $A$ such that $C = AB$ for some $B$, can we find $B$? Or in other words, if $C$ is indeed composed of two parts, and we know one part, what is the other? 
    This problem is particularly interesting, as the FDS $B$ may not be unique: there exist many examples of FDSs $A,B,D$ such that $AB = AD$. 

    \item The third problem is $k$-th root: given an FDS $A$ and an integer $k$, is there $B$ such that $B^k = A$, and how many such roots exist? Until now, very little is known about this problem; for instance there was no result asserting that the solution $B$ should be unique.
\end{enumerate}

In this paper, we establish important connections between FDSs and infinite, periodic trees. In particular, we introduce the unrolling of an FDS, which can be viewed as a space-time expansion of the system. The unrolling preserves all the information about the transient dynamics of an FDS, and preserves the product operation. However, the product on trees (and in particular, on unrollings of FDSs) is much better behaved than its counterpart for FDSs and hence allows us to prove our main results.

This paper makes four main contributions towards the understanding of the three problems listed above.
\begin{enumerate}
    \item \label{itemConnected}
    An FDS is connected if its graph is connected; in other words, it has only one periodic cycle. We first prove a fundamental property of connected FDSs. For any FDS $A$, if $X$ and $Y$ are connected and $AX = AY$, then $X = Y$. Intuitively, this means that division is unambiguous when we know the quotient is connected.

    \item \label{itemCancellative}
    Intuitively, a cancellative FDS is one those that can be unambiguously divided by. Formally, $A$ is cancellative if $AB = AC$ implies $B = C$. Our first and major result is the characterisation of cancellative FDSs: they are exactly those with a fixpoint. This result is close that Theorem 8 of \cite{LovaszCancellation1971}, but not equivalent, since Lovász's paper studies cancellativity in the semiring of digraphs (and thus, there could be an FDS that is cancellative on $\DD$ but not as a general digraph).
    
    \item \label{itemAlgorithm}
    Our proof methods involve working with (possibly infinite) trees, and going back and forth between FDSs and trees. As a bi-product, we obtain an algorithm for division of dendrons. That is, given two dendrons $A$ and $B$, the algorithm determines the dendron $C$ such that $A = BC$ or returns a failure if no such dendron exist. It is easily shown that this algorithm runs in time polynomial in the size of the input.
    
    \item \label{itemRoots}
    Our result on cancellative FDSs has an important consequence: many polynomials in $\DD[X]$ are injective. We then prove that the polynomial $P(X) = X^k$ is injective, i.e. if an FDS has a $k$-th root then it is unique.
    
    \item \label{itemLDK}
    Throughout the paper, we investigate the structure of division and factorisation in dendrons. We further this investigation by exhibiting large monoids of dendrons with unique factorisation.
\end{enumerate}

While writing up this paper, we have discovered the related paper \cite{DFPR22}. This work and \cite{DFPR22} have two main similarities: both have independently introduced the unrolling construction, and both have proved the cancellative nature of the product on infinite trees (\Cref{InfTreeDiv} in this work, Theorem 3.3 in \cite{DFPR22}). However, we would like to stress the significant differences between these two papers. First, the respective proofs of the result mentioned above are completely different: ours is based on a lexicographic order on trees, while theirs is based on counting tree homomorphisms. Second, and more importantly, both papers consider completely different problems about FDSs. As such, the last four main contributions of this work (items \ref{itemCancellative} to \ref{itemLDK} in the list above) are novel and do not appear in the literature so far. Third, \cite{DFPR22} proposes the following conjecture (Conjecture 3.1): let $A$ and $B$ be two connected FDSs, then for all FDSs $X$ and $Y$, if $AX = B$ and $AY = B$, then $X = Y$. Our first main contribution (item \ref{itemConnected} in the list above) was added once we were aware of \cite{DFPR22}; it is actually a more general result than their conjecture.

The rest of the paper is organised as follows. \Cref{sec:Defns} introduces all the necessary
definitions to work on the semiring of FDSs. Then,
\Cref{sec:Cancellative} shows that the cancellative elements of the
semiring of finite dynamical systems are exactly those with a
fixpoint. From this, we give in \Cref{sec:Algos} a
polynomial-time algorithm for division in dendrons. \Cref{kroot} then proves the unicity of $k$-th roots. Then
\Cref{sec:LD_K} constructs a class of monoids with unique
factorisation on each of them. Finally, some avenues for further work are proposed in
\Cref{sec:Directions}.

\section{General definitions}\label{sec:Defns}

A finite dynamical system (FDS) is a function from a finite set into itself. We denote $\DD$ the set of all FDSs.
Given an FDS $A$, we denote $S_A$ the finite set on which it acts.

Given two FDSs $A$ and $B$, we can assume that $S_A\cap S_B=\es$ (if
that is not the case, we can simply rename the elements of one of
those sets). Then, we define their \textit{sum} as follows:

$$
\begin{array}{llll}
  A+B: & S_A\sqcup S_B & \rightarrow & S_A\sqcup S_B \\
  & x             & \mapsto     &
  \begin{cases}
    A(x) & \text{ if $x\in S_A$} \\
    B(x) & \text{ otherwise.}
  \end{cases}
\end{array}
$$

Given two FDSs $A$ and $B$, we define their \textit{product} as
follows:

  $$
  \begin{array}{cccc}
    AB: & S_A\times S_B & \rightarrow & S_A\times S_B \\
         & (a,b)         & \mapsto     & (A(a), B(b)).
  \end{array}
  $$
  
Defining the size of an FDS $A$ as $|A|=|S_A|$, we see that:
$|A+B|=|A|+|B|$ and $|AB|=|A||B|$.
  
When multiplying two FDSs $A$ and $B$ (for example, in \Cref{FDSmul}),
we get $AB$ along with a labelling of the states of $AB$ by pairs of
states of $A$ and of $B$. That is, we get an isomorphism $S_{AB}\simeq
S_A\times S_B$ that respects the product structure. However, we shall
consider FDSs up to isomorphism, hence we do not get this labelling
along with the FDS in general. The problem of factorising an FDS $C$,
for example, just means labelling the states $S_C$ with $S_A\times
S_B$, where $A,B\in\DD$, such that this labelling respects the product
$AB=C$. A formalization of this idea of labelling the states of a
product with the Cartesian product of the state sets of its factors is
provided below:

\begin{defn}
  Given a sequence $(A_i)_{i\in I}\in \DD^I$ for some finite set $I$ and a product $B=\prod_{i\in I} A_i$ we say that the function
  $\phi:S_B\mapsto \prod_{i\in I}S_{A_i}$ is a \textit{product
    isomorphism for the product $B=\prod_{i\in I}A_i$} if:
  \begin{enumerate}
  \item it is a bijection, and
  \item for any sequence of states $(s_i)_{i\in I}\in \prod_{i\in
    I}S_{A_i}$, we have: $B(\phi^{-1}((s_i)_{i\in I})) =
    \phi^{-1}((A_i(s_i))_{i\in I})$.
  \end{enumerate}
\end{defn}

We remark that computing the product gives such a product isomorphism.

\begin{prop}[\cite{PolEqFDS}]
  The set of FDSs, with the above sum and product, forms a semiring \cite{HW98},
  with additive identity the empty function and multiplicative
  identity the function $1:\{1\}\rightarrow\{1\}, 1\mapsto 1$.
\end{prop}

For FDSs, we can adopt a graph-theoretical point of view, by
associating to an FDS $A\in \DD$ an oriented graph
$\mathcal{G}_A=(V,E)$ where $V=S_A$ and $E=\{(x,y)\in S_A^2:y=A(x)\}$.
Then, for $A,B\in\DD$, $\mathcal{G}_{A+B}$ is the disjoint union of
$\Gg_A$ and $\Gg_B$, and $\mathcal{G}_{AB}$ is the direct product
$\Gg_A\times G_B$ (see the corresponding section in
\cite{HandbookProds}). In the following, we will often identify an FDS
and its graph, and thus, implicitly quotient $\DD$ by graph
isomorphism, that is, we consider that $A=B$ if and only if $\Gg_A$
and $\Gg_B$ are isomorphic, as $A$ and $B$ have the same dynamics in
this case.

\begin{defn}
    Given $A,B\in\DD$, we say that $B$ is a sub-FDS of $A$ if $\Gg_B$
    is a subgraph of $\Gg_A$.
\end{defn}

Since FDSs take their values in a finite set, the structure of their
graphs is simple: they consist of some cycles, on the states of which,
trees (with arrows going upwards, towards the root) are connected, as
in the example of \Cref{FDSmul}. This leads us to several definitions
that are useful to study FDSs.

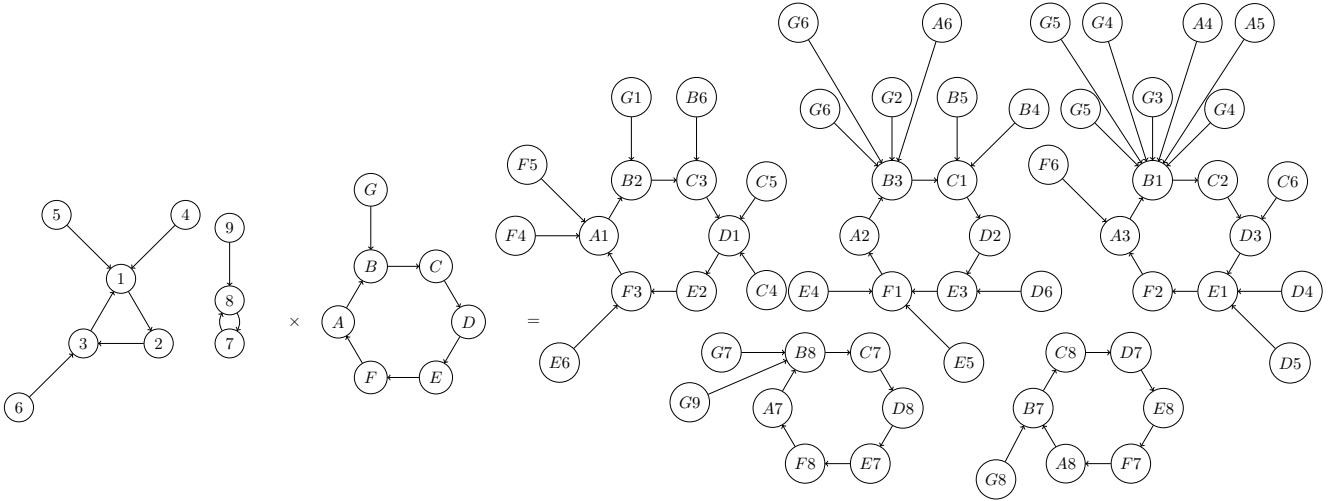
\begin{figure}
  \resizebox{\textwidth}{!}{
    \begin{tikzpicture}
      \graph[nodes={circle,draw}] {
        subgraph C_n [n=3, clockwise, ->]
      };
      \node[draw,circle] (4) [above right=of 1] {$4$};
      \node[draw,circle] (5) [above left=of 1] {$5$};
      \node[draw,circle] (6) [below left=of 3] {$6$};
      \draw[->] (4) -- (1);
      \draw[->] (5) -- (1);
      \draw[->] (6) -- (3);

      \tikzxshift{2.5cm}
      
      \node[draw,circle] (7) at (0,-0.5) {$7$};
      \node[draw,circle] (8) at (0,0.5) {$8$};
      \node[draw,circle] (9) [above=of 8] {$9$};
      \draw[->] (7) edge[bend left] (8);
      \draw[->] (8) edge[bend left] (7);
      \draw[->] (9) -- (8);
      
      \tikzxshift{1.5cm}
      \node {$\times$};
      \tikzxshift{2.5cm}

      \node[draw,circle] (A) at (-1.5,0) {$A$};
      \node[draw,circle] (B) at (-0.5*1.5,1.5*0.86) {$B$};
      \node[draw,circle] (C) at (0.5*1.5,1.5*0.86) {$C$};
      \node[draw,circle] (D) at (1.5,0) {$D$};
      \node[draw,circle] (E) at (0.5*1.5,-1.5*0.86) {$E$};
      \node[draw,circle] (F) at (-0.5*1.5,-1.5*0.86) {$F$};
      \node[draw,circle] (G) [above=of B] {$G$};
      \draw[->] (G) -- (B);
      \draw[->] (B) -- (C);
      \draw[->] (C) -- (D);
      \draw[->] (D) -- (E);
      \draw[->] (E) -- (F);
      \draw[->] (F) -- (A);
      \draw[->] (A) -- (B);

      \tikzxshift{3cm}
      \node {$=$};
      \tikzxshift{3cm}
      \tikzyshift{2cm}

      \node[draw,circle] (A1) at (-1.5,0) {$A1$};
      \node[draw,circle] (B2) at (-0.5*1.5,1.5*0.86) {$B2$};
      \node[draw,circle] (C3) at (0.5*1.5,1.5*0.86) {$C3$};
      \node[draw,circle] (D1) at (1.5,0) {$D1$};
      \node[draw,circle] (E2) at (0.5*1.5,-1.5*0.86) {$E2$};
      \node[draw,circle] (F3) at (-0.5*1.5,-1.5*0.86) {$F3$};
      \draw[->] (B2) -- (C3);
      \draw[->] (C3) -- (D1);
      \draw[->] (D1) -- (E2);
      \draw[->] (E2) -- (F3);
      \draw[->] (F3) -- (A1);
      \draw[->] (A1) -- (B2);

      \tikzxshift{6cm}

      \node[draw,circle] (A2) at (-1.5,0) {$A2$};
      \node[draw,circle] (B3) at (-0.5*1.5,1.5*0.86) {$B3$};
      \node[draw,circle] (C1) at (0.5*1.5,1.5*0.86) {$C1$};
      \node[draw,circle] (D2) at (1.5,0) {$D2$};
      \node[draw,circle] (E3) at (0.5*1.5,-1.5*0.86) {$E3$};
      \node[draw,circle] (F1) at (-0.5*1.5,-1.5*0.86) {$F1$};
      \draw[->] (A2) -- (B3);
      \draw[->] (B3) -- (C1);
      \draw[->] (C1) -- (D2);
      \draw[->] (D2) -- (E3);
      \draw[->] (E3) -- (F1);
      \draw[->] (F1) -- (A2);

      \tikzxshift{6cm}
      
      \node[draw,circle] (A3) at (-1.5,0) {$A3$};
      \node[draw,circle] (B1) at (-0.5*1.5,1.5*0.86) {$B1$};
      \node[draw,circle] (C2) at (0.5*1.5,1.5*0.86) {$C2$};
      \node[draw,circle] (D3) at (1.5,0) {$D3$};
      \node[draw,circle] (E1) at (0.5*1.5,-1.5*0.86) {$E1$};
      \node[draw,circle] (F2) at (-0.5*1.5,-1.5*0.86) {$F2$};
      \draw[->] (C2) -- (D3);
      \draw[->] (D3) -- (E1);
      \draw[->] (E1) -- (F2);
      \draw[->] (F2) -- (A3);
      \draw[->] (A3) -- (B1);
      \draw[->] (B1) -- (C2);

      \tikzyshift{-4cm}
      \tikzxshift{-8cm}

      \node[draw,circle] (A7) at (-1.5,0) {$A7$};
      \node[draw,circle] (B8) at (-0.5*1.5,1.5*0.86) {$B8$};
      \node[draw,circle] (C7) at (0.5*1.5,1.5*0.86) {$C7$};
      \node[draw,circle] (D8) at (1.5,0) {$D8$};
      \node[draw,circle] (E7) at (0.5*1.5,-1.5*0.86) {$E7$};
      \node[draw,circle] (F8) at (-0.5*1.5,-1.5*0.86) {$F8$};
      \draw[->] (C7) -- (D8);
      \draw[->] (D8) -- (E7);
      \draw[->] (E7) -- (F8);
      \draw[->] (F8) -- (A7);
      \draw[->] (A7) -- (B8);
      \draw[->] (B8) -- (C7);

      \tikzxshift{6cm}
      
      \node[draw,circle] (B7) at (-1.5,0) {$B7$};
      \node[draw,circle] (C8) at (-0.5*1.5,1.5*0.86) {$C8$};
      \node[draw,circle] (D7) at (0.5*1.5,1.5*0.86) {$D7$};
      \node[draw,circle] (E8) at (1.5,0) {$E8$};
      \node[draw,circle] (F7) at (0.5*1.5,-1.5*0.86) {$F7$};
      \node[draw,circle] (A8) at (-0.5*1.5,-1.5*0.86) {$A8$};
      \draw[->] (B7) -- (C8);
      \draw[->] (C8) -- (D7);
      \draw[->] (D7) -- (E8);
      \draw[->] (E8) -- (F7);
      \draw[->] (F7) -- (A8);
      \draw[->] (A8) -- (B7);

      \node[draw,circle] (G1) [above=of B2] {$G1$};
      \draw[->] (G1) -- (B2);
      \node[draw,circle] (G2) [above=of B3] {$G2$};
      \draw[->] (G2) -- (B3);
      \node[draw,circle] (G3) [above=of B1] {$G3$};
      \draw[->] (G3) -- (B1);
      \node[draw,circle] (G4) [above right=of B1] {$G4$};
      \draw[->] (G4) -- (B1);
      \node[draw,circle] (G5) [above left=of B1] {$G5$};
      \draw[->] (G5) -- (B1);
      \node[draw,circle] (G6) [above left=of B3] {$G6$};
      \draw[->] (G6) -- (B3);
      \node[draw,circle] (G7) [left=of B8] {$G7$};
      \draw[->] (G7) -- (B8);
      \node[draw,circle] (G8) [below left=1cm and 0.2cm of B7] {$G8$};
      \draw[->] (G8) -- (B7);
      \node[draw,circle] (G9) [below left=0.5cm and 2cm of B8] {$G9$};
      \draw[->] (G9) -- (B8);

      \node[draw,circle] (A4) [above right=3cm and 0.5cm of B1] {$A4$};
      \draw[->] (A4) -- (B1);
      \node[draw,circle] (B4) [above right=of C1] {$B4$};
      \draw[->] (B4) -- (C1);
      \node[draw,circle] (C4) [below right=0.6cm and 0.2cm of D1] {$C4$};
      \draw[->] (C4) -- (D1);
      \node[draw,circle] (D4) [right=of E1] {$D4$};
      \draw[->] (D4) -- (E1);
      \node[draw,circle] (E4) [left=of F1] {$E4$};
      \draw[->] (E4) -- (F1);
      \node[draw,circle] (F4) [left=of A1] {$F4$};
      \draw[->] (F4) -- (A1);
      \node[draw,circle] (G4) [above left=3cm and 0.5cm of B1] {$G4$};
      \draw[->] (G4) -- (B1);

      \node[draw,circle] (A5) [above right=3cm and 1.7cm of B1] {$A5$};
      \draw[->] (A5) -- (B1);
      \node[draw,circle] (B5) [above=of C1] {$B5$};
      \draw[->] (B5) -- (C1);
      \node[draw,circle] (C5) [above right=0.6cm and 0.2cm of D1] {$C5$};
      \draw[->] (C5) -- (D1);
      \node[draw,circle] (D5) [below right=of E1] {$D5$};
      \draw[->] (D5) -- (E1);
      \node[draw,circle] (E5) [below right=of F1] {$E5$};
      \draw[->] (E5) -- (F1);
      \node[draw,circle] (F5) [above left=of A1] {$F5$};
      \draw[->] (F5) -- (A1);
      \node[draw,circle] (G5) [above left=3cm and 1.7cm of B1] {$G5$};
      \draw[->] (G5) -- (B1);

      \node[draw,circle] (A6) [above right=3cm and 0.5cm of B3] {$A6$};
      \draw[->] (A6) -- (B3);
      \node[draw,circle] (B6) [above=of C3] {$B6$};
      \draw[->] (B6) -- (C3);
      \node[draw,circle] (C6) [above right=0.6cm and 0.2cm of D3] {$C6$};
      \draw[->] (C6) -- (D3);
      \node[draw,circle] (D6) [right=of E3] {$D6$};
      \draw[->] (D6) -- (E3);
      \node[draw,circle] (E6) [below left=of F3] {$E6$};
      \draw[->] (E6) -- (F3);
      \node[draw,circle] (F6) [above left=of A3] {$F6$};
      \draw[->] (F6) -- (A3);
      \node[draw,circle] (G6) [above left=3cm and 1.5cm of B3] {$G6$};
      \draw[->] (G6) -- (B3);
    \end{tikzpicture}
  }
  \caption{Product of two FDSs.}
  \label{FDSmul}
\end{figure}

\begin{defn}
  Let $A\in\DD$. A state $s\in S_A$ is said to be a \textit{cycle
    state} if it is on a cycle of $\Gg_A$, or, equivalently, if there
  exists $n>0$ such that $A^n(s)=s$. We denote $S^C_A$ the set of
  cycle states of $A$. Otherwise, $s$ is said to be a \textit{tree
    state}.

  We define a function $\depth_A: S_A\rightarrow\N$ that gives the
  \textit{depth} of any state of $A$, and is defined recursively as
  follows:
  \begin{eqnarray*}
    \forall s\in S^C_A, &&\depth_A(s)=0 \\
    \forall s\in S_A\setminus S^C_A, &&\depth_A(s) = \depth_A(A(s))+1.
  \end{eqnarray*}

  Furthermore, for any $k\in\N$, we define the \textit{truncature of
    $A$ at depth $k$}, denoted $[A]_k$, as the sub-FDS of $A$ which
  contains all the states of $A$ at depth at most $k$.
\end{defn}

A very useful and simple result is the following:

\begin{lem}
  For any $A,B\in\DD$ and $k\in\N$, $[AB]_k=[A]_k[B]_k$.
\end{lem}

Of particular interest are the FDSs we call \textit{dendrons}, that
is, connected FDSs (\textit{i.e.} with a connected graph) with a
fixpoint. Those FDS can be seen as rooted trees with arrows pointing
towards the root, with a loop on the root. We denote $\DD_D$ the set
of dendrons (remark that it is not a semiring, since the sum of two
dendrons is not a dendron).

Let's now focus on those two types of parts of FDSs: trees (which,
when summed, form forests) and cycles (which, when summed, form
permutations).

\subsection{Forests}

We introduce forests as a way to have a product between FDSs that has
an inductive definition that works level by level. In FDSs, the state
set of a product is the Cartesian product of the state sets of the
factors. This makes the identification of states of an unlabelled FDS
difficult. For forests, the pairs of states which end up in the
product are those of even depth. Finally, \Cref{ProdBot} is the reason
forests are useful: their product is compatible with that of FDSs.

\begin{defn}
  By \textit{tree}, we shall mean an in-tree \cite[p.21]{BG09a}, i.e. an oriented connected acyclic graph
  with a special vertex called its \textit{root}, such that every edge is oriented towards the root. The trees we consider may be infinite, but the degree of each vertex shall always be finite.
  
  % A tree can be infinite, but each vertex must have   a finite degree. 
  We denote the root of a tree $\bott T$ as
  $\troot(\bott T)$.

  A \textit{forest} is a disjoint union of trees. The set of forests
  is denoted $\FF$, and that of trees is denoted $\FF_T$. In the
  following, we denote forests in bold face to distinguish them from
  FDSs.

  If $\bott T\in\FF_T$, and if there is a single infinite path
  starting from the root of $\bott T$, we can extract the sequence
  $tseq(\bott T)$ of trees anchored on that path. If this sequence is
  periodic, we say that $\bott T$ is \textit{periodic}, and that
  $\bott T$ is of \textit{tree period} the period of the sequence. We
  denote the set of periodic trees as $\FF_P$.
\end{defn}

We consider trees as dendrons which have had their fixpoint
transformed into a sink, and extend the notations from dendrons
whenever they make sense. In particular, we denote $S_{\bott A}$ the
set of vertices of the forest $\bott{A}$. Moreover, the parent of a
vertex $x\in S_{\bott A}$ is denoted $\bott{A}(x)$. For an FDS $A \in
\DD$ and a state $s\in S_A$, we say that $\bott T$ is the tree
\textit{anchored on} $s$ if the tree of the tree state predecessors of
$s$ in the graph is $\bott T$; we naturally extend this notation to any forest $\bott A$. By convention, the depth of an infinite dendron is $\infty$, while the depth of an empty dendron is $-1$.

Given a tree $\bott T$, we define $\Dd(\bott T)$ to be the multiset
containing the subtrees anchored on the children of the root of
$\bott T$.

Now, we define a sum and a product operation on forests in order to
endow the set of forests with a semiring structure.

The \textit{sum} of two forests $\bott A,\bott B$ (for which we can
assume $S_{\bott A}\cap S_{\bott B}=\es$) is the forest $\bott C$
defined as the disjoint union of the graphs $\bott A$ and $\bott B$.

Let $\bott A, \bott B\in\FF_T$. Then the \textit{product} of $\bott A$
and $\bott B$ is $\bott{A}\bott B=(V,A)$ with
\begin{align*}
    V = S_{\bott A\bott B} &= \{(a,b)\in S_{\bott A}\times S_{\bott B}: \depth_{\bott A}(a)=\depth_{\bott B}(b)\},\\
    A &= \{((a,b), (\bott{A}(a),\bott{B}(b))): (a,b)\in S_{\bott A\bott B}\}. 
\end{align*}
This
product is almost the same as that on FDSs but here only states of
same depth get multiplied together.

It will often prove useful to use multisets with the following
product. Given two multisets $\Aa$ and $\mathcal{B}$, their
\textit{product} $\Aa\mathcal{B}$ is
$\{\{ab:a\in\Aa,b\in\mathcal{B}\}\}$.

The following lemma explains why trees are interesting: the product is
done level by level. Moreover, the root does not behave differently
than the other states (as it does on dendrons), which means that this
product is much easier to work with.

\begin{lem}\label{LevelByLevelProduct}
  If $\bott A,\bott B,\in\FF_T$ are finite, then: $$\Dd(\bott A\bott
  B)=\Dd(\bott A)\Dd(\bott B)=\{\{\bott T\times\bott{T'}:\bott T\in
  \Dd(\bott A), \bott{T'}\in\Dd(\bott B)\}\}.$$
\end{lem}
\begin{proof}
  The proof is by induction on the depths of $\bott A$ and $\bott B$.
  The case for trees with depth $\leq 1$ is trivial. The depth $1$
  vertices of $\bott A\bott B$ form the set $\{(a,b)\in S_{\bott
    A}\times S_{\bott B}: \depth_{\bott A}(a)=\depth_{\bott
    B}(b)=1\}$. We simply show that the tree $\bott T_{(a,b)}$
  anchored on $(a,b)$ in $\bott A\bott B$ is the product of the tree $\bott T_{a}$
  anchored on $a$ in $\bott A$ with the tree $\bott T_{b}$ anchored
  on $b$ in $\bott B$. By induction, we know that $\Dd(\bott
  T_{(a,b)})=\Dd(\bott T_{a})\Dd(\bott T_{b})$, so $\bott
  T_{(a,b)}=\bott T_{a}\bott T_{b}$. This concludes.
\end{proof}

It is easy to verify that the set of forests becomes a semiring with these operations:

\begin{lem}
  The set $\FF$ of forests becomes a semiring when endowed with the
  sum and product defined above. Its additive identity is $\bott 0$, the empty
  tree with $(V=\es, A=\es)$, while its multiplicative identity is the
  rooted infinite directed path $\bott P_\infty$ with $(V=\N,
  A=\{(n+1, n)|n\in\N\})$.
\end{lem}

A straightforward inductive proof gives the following lemma:

\begin{lem}\label{DepthBotTreeState}
  If $\bott A, \bott B\in\FF$, then for $a\in S_{\bott A}, b\in
  S_{\bott B}$ such that $(a,b)\in S_{\bott A\bott B}$, we have
  $\depth_{\bott A\bott B}((a,b))=\depth_{\bott A}(a)=\depth_{\bott
    B}(b)$.
\end{lem}

\begin{lem}\label{DepthTrees}
  If $\bott A, \bott B\in\FF$, then $\depth(\bott A\bott
  B)=\min(\depth(\bott A), \depth(\bott B))$.
\end{lem}
\begin{proof}
  We have:
  $$S_{\bott A\bott B}=\{(a,b)\in S_{\bott A}\times S_{\bott B}:
  \depth_{\bott A}(a)=\depth_{\bott B}(b)\}.$$

  From \Cref{DepthBotTreeState}, a state $(a,b)\in S_{\bott A}\times
  S_{\bott B}$ has depth at most $\min(\depth_{\bott
    A}(a),\depth_{\bott B})$. Moreover, if $k=\min(\depth(\bott A),
  \depth(\bott B))$ and we let $a\in S_{\bott A},b\in S_{\bott B}$ two
  states that have both depth $k$ in their respective trees, then
  $(a,b)$ has depth $k$ too.
\end{proof}

\subsection{Permutations}

For every $k\geq 1$, we denote $C_k$ the \textit{cycle of length $k$}
defined as the FDS whose graph is the directed cycle of length $k$. We
say that $A\in\DD$ is a \textit{permutation} if the function $A$ is
bijective. In that case, all the states of $A$ are cycle states. We
denote the semiring of permutations $\DD_P$ (it has multiplicative
identity $C_1$ and additive identity the empty function). In
particular, for any $A\in\DD$, $[A]_0\in \DD_P$. 

 We introduce two shortened notations: $a\vee
b=\text{lcm}(a,b)$ and $a\wedge b=\text{gcd}(a,b)$. In \cite{PolEqFDS}, the following very useful and simple result is
 proven:
 
\begin{lem}\label{LemCycleProduct}
  $C_a\times C_b=(a\wedge b)C_{a\vee b}$.
\end{lem}

We now extend it to arbitrary products of cycles. For any multiset $J$ of positive integers, we denote $\bigmeet J = \bigmeet_{j \in J} j$ and $\bigjoin J = \bigjoin_{j \in J} j$; if $J$ is empty, then those terms are equal to $1$.

\begin{lem} \label{LemDelta}
Let $J$ be a multiset of positive integers. Then $\prod_{j \in J} C_j = \delta_J C_{\bigjoin J}$, where $\delta_J$ is recursively defined as $\delta_\emptyset = 1$ and for any $a \in \N$
\[
    \delta_{ J \cup \{a\} } = (a \meet \bigjoin J) \delta_J.
\]
\end{lem}

\begin{proof}
 The proof is by induction on the cardinality of $J$. The result is clear when $J$ is empty. Assume it is true for $J$, and let $a \in \N$. Then
 \[
    \prod_{ j \in J \cup \{a\} } C_j = C_a \prod_{j \in J} C_j = \delta_J  C_a C_{\bigjoin J} = \delta_J ( a \meet \bigjoin J ) C_{\bigjoin J \cup \{a\}}.
 \]
\end{proof}

Given a permutation
$A\in \DD_P$, such that the length of each of its cycles is a multiple
of some $k\in\N$, and $k$ trees $\bott T_0, \dots, \bott
T_{k-1}\in\FF_T$, we denote $A(\bott T_0, \dots, \bott T_{k-1})$ the FDS
obtained by taking each cycle of $A$, traversing it by following the
arrows, and anchoring on the $i$-th state encountered the dendron
$T_{i\mod k}$. This is pictured in \Cref{CycleTreezation}.

We use the following notation: for $A\in \DD$, and for all $i\in\N$,
we denote $\lambda_i^A$ the number of cycles of length $i$ in $A$.

Finally, given an FDS $A\in\DD$, and a set $L\subs\N$, we
define the \textit{$L$-support} of $A$, denoted $\supp_L(A)$, as the FDS
made of the connected components of $A$ with cycle size in $L$.

The following two results will prove useful to understand the product
of a permutation with a dendron.

\begin{lem}
    For any $\ell, k \geq 1$ and trees $\bott T_1, \dots, \bott T_k\in\FF_T$,
    $C_k(\bott T_1, \dots, \bott T_k)\times C_\ell=(C_kC_\ell)(\bott T_1, \dots, \bott T_k)$.
\end{lem}
\begin{proof}
  Take a product isomorphism for the product $C_k\times
  C_\ell=(k\wedge\ell)C_{k\vee \ell}$, and write $S_{(k\wedge
    \ell)}C_{k\vee\ell}\simeq S_{C_k}\times S_{C_\ell}$ accordingly.
  Then, take $(i,j)\in S_{(k\wedge \ell)C_{k\vee\ell}}$. Let's show
  that the tree that is anchored on this state in $S_{C_k(\bott T_1, \dots,
    \bott T_k)\times C_\ell}$ is the tree that is anchored on $i$ in
  $C_k(\bott T_1, \dots, \bott T_k)$, say $\bott T_i$. Indeed, since $C_\ell$ has no tree
  states, the tree states over $(i,j)$ have a first component with is
  a tree state, and a second one which is a cycle state. But since
  each state of $C_\ell$ has exactly one predecessor, and the tree
  anchored on $(i,j)$ is indeed $\bott T_i$. This proves the result.
\end{proof}

\begin{cor}\label{CorMultPermTrees}
  For any $A\in\DD_P$ and trees $\bott T_1, \dots, \bott T_k\in\FF_T$, $C_k(\bott T_1,
  \dots, \bott T_k)\times A = (C_kA)(\bott T_1, \dots, \bott T_k)$.
\end{cor}
\begin{proof}
  Let's write $A=\sum_{i\in\N}\lambda^A_iC_i$. Then, 
  \[
  C_k(\bott T_1, \dots,
  \bott T_k)\times A=\sum_{i\in\N}\lambda^A_i (C_kC_i)(\bott T_1, \dots, \bott T_k) =
  (C_kA)(\bott T_1, \dots, \bott T_k)
  \]
  from the previous lemma.
\end{proof}

\begin{figure}
  \begin{subfigure}{0.3\textwidth}
    \begin{center}
      \begin{tikzpicture}
        \node[draw,circle] (1) {\phantom{$1$}};
        \node[draw,circle] (2) [above right=of 1] {\phantom{$2$}};
        \node[draw,circle] (3) [above left=of 1] {\phantom{$3$}};
        \draw[->] (2) -- (1);
        \draw[->] (3) -- (1);
      \end{tikzpicture}
    \end{center}
    \caption{$\bott T_1$.}
  \end{subfigure}
  \begin{subfigure}{0.3\textwidth}
    \begin{center}
      \begin{tikzpicture}
        \node[draw,circle] (1) {\phantom{$1$}};
        \node[draw,circle] (2) [above=of 1] {\phantom{$2$}};
        \node[draw,circle] (3) [above=of 2] {\phantom{$3$}};
        \draw[->] (2) -- (1);
        \draw[->] (3) -- (2);
      \end{tikzpicture}
    \end{center}
    \caption{$\bott T_2$.}
  \end{subfigure}
  \begin{subfigure}{0.3\textwidth}
    \resizebox{1\textwidth}{!}{
    \begin{tikzpicture}
      \node[draw,circle] (1) {\phantom{$1$}};
      \node[draw,circle] (2) [below=of 1] {\phantom{$2$}};
      \draw[->] (1) edge[bend left] (2);
      \draw[->] (2) edge[bend left] (1);

      \node[draw,circle] (1A1) [above left=of 1] {\phantom{$1$}};
      \node[draw,circle] (1A2) [above right=of 1] {\phantom{$1$}};
      \draw[->] (1A1) -- (1);
      \draw[->] (1A2) -- (1);

      \node[draw,circle] (1B1) [below=of 2] {\phantom{$1$}};
      \node[draw,circle] (1B2) [below=of 1B1] {\phantom{$1$}};
      \draw[->] (1B1) -- (2);
      \draw[->] (1B2) -- (1B1);
      
      \tikzxshift{4cm}
      \node[draw,circle] (3) {\phantom{$3$}};
      \node[draw,circle] (4) [right=of 3] {\phantom{$4$}};
      \node[draw,circle] (5) [below=of 4] {\phantom{$5$}};
      \node[draw,circle] (6) [left=of 5] {\phantom{$6$}};
      \draw[->] (3) -- (4);
      \draw[->] (4) -- (5);
      \draw[->] (5) -- (6);
      \draw[->] (6) -- (3);

      \node[draw,circle] (5A1) [below left=of 5] {\phantom{$1$}};
      \node[draw,circle] (5A2) [below right=of 5] {\phantom{$1$}};
      \draw[->] (5A1) -- (5);
      \draw[->] (5A2) -- (5);

      \node[draw,circle] (3A1) [above left=of 3] {\phantom{$1$}};
      \node[draw,circle] (3A2) [above right=of 3] {\phantom{$1$}};
      \draw[->] (3A1) -- (3);
      \draw[->] (3A2) -- (3);

      \node[draw,circle] (2B1) [below left=of 6] {\phantom{$1$}};
      \node[draw,circle] (2B2) [below=of 2B1] {\phantom{$1$}};
      \draw[->] (2B1) -- (6);
      \draw[->] (2B2) -- (2B1);
      
      \node[draw,circle] (4B1) [above right=of 4] {\phantom{$1$}};
      \node[draw,circle] (4B2) [above=of 4B1] {\phantom{$1$}};
      \draw[->] (4B1) -- (4);
      \draw[->] (4B2) -- (4B1);
    \end{tikzpicture}
    }
    \caption{$A(\bott T_1, \bott T_2)$.}
  \end{subfigure}

  \caption{The FDS $A(\bott T_1, \bott T_2)$ for $A=C_2+C_4$, and two
    trees $\bott T_1, \bott T_2$.}
  \label{CycleTreezation}
\end{figure}

\section{Cancellative finite dynamical systems} \label{sec:Cancellative}

It is known that the division operation can sometimes not yield a
unique result; a well-known example is: $C_2^2=2C_2$. We can also show
that if we have $AB=AC$ for $A,B,C\in\DD$, even setting $[B]_0=[C]_0$
does not guarantee that $B=C$, since, given two different trees $\bott
T_1$ and $\bott T_2$, we have the identity
$C_2(2\bott T_1+C_2(\bott T_2))=C_2(C_2(\bott T_1)+2\bott T_2)$. We therefore consider the elements for which division is unambiguous, defined as follows.

\begin{defn}[Cancellative element]
  An FDS $A\in\DD$ is said \textit{cancellative} if for all
  $B,C\in\DD$, $AB=AC\implies B=C$.
\end{defn}

In this section, we prove that an FDS is cancellative if and only if
it has a fixpoint. We approach this theorem in steps. First, we
introduce an order on trees, based on a code for trees. Then, we move
on to show that we can transform FDSs into forests in a way that works
well with both the product on forests and on FDSs. Finally, we show
that result.

\subsection{Order on trees}

It will prove very useful to
have a total order on finite trees that is compatible with the
product. That is, if $\bott T_1,\bott T_2,\bott T_3,\bott
T_4\in\FF_T$, and $\bott T_1 < \bott T_2$ and $\bott T_3\leq\bott T_4$,
we want to have $\bott T_1\bott T_3<\bott T_2\bott T_4$. This will be
guaranteed by \Cref{CompleteCodeOrder}.

To do so, we define a code $\Cc_f$ from finite trees to $\N^*$ (the
set of finite sequences of nonnegative integers), and we say that
$\bott T_1<\bott T_2\iff \Cc_f(\bott T_1)<_{\text{lex}}\Cc_f(\bott
T_2)$ (where $<_{\text{lex}}$ is the lexicographical order).

The code is computed as follows, using two mutually recursive
functions. We consider for a moment that trees are ordered: the
children of a node are stored in an ordered list, say, from left to
right. That is, for a tree $\bott T$, $\Dd(\bott T)$ is now a tuple
rather than a multiset. Then, we define a procedure $collect$ which
takes a finite tree, sorts it (using the function $sort$ defined
below), and then traverses level by level, following the
order of the predecessors, starting from depth $0$, and outputs a
tuple of the number of predecessors of each node encountered.

We also define a procedure $sort$ that takes a finite tree $\bott T$,
begins by calling $collect$ on each of the subtrees anchored on direct
predecessors of the root, and then order those predecessors from left
to right by increasing return value of $collect$. Finally, $\Cc_f(\bott
T)=collect(\bott T)$.

A pseudocode implementation of $sort$ and $collect$ is found in \Cref{SortAndCollect}.

\begin{figure}
  \begin{subfigure}[B]{0.5\textwidth}
    \begin{center}
      \begin{procedure}[H]
        \If{$|\bott T|>1$}{
          $\bott T\leftarrow sort(\bott T)$\;
        }
        $t=[]$\;
        \For{$i\in\iitv{0,\depth(\bott T)}$}{
          \For{$v$ in $\bott T$'s depth $i$, from left to right}{
            $d\leftarrow$ number of children of $v$\;
            $t\leftarrow t :: d$\;
          }
        }
        \Return $t$\;
      \end{procedure}
    \end{center}
    \caption{$collect(\bott T)$ ($::$ is the concatenation operator)}
  \end{subfigure}
  \begin{subfigure}[B]{0.5\textwidth}
    \begin{center}
      \begin{procedure}[H]
        $(T_1, \dots, T_n) := \Dd(\bott T)$\;
        \For{$i\in\iitv{1,n}$}{
          $c_i\leftarrow collect(T_i)$\;
        }
        $(U_1, \dots, U_n)\leftarrow$ sort $(T_1, \dots, T_n)$ by increasing
        $(c_1, \dots, c_n)$\;
        Let $\bott{T'}$ such that $\Dd(\bott{T'})=(U_1, \dots, U_n)$\;
        \Return $\bott{T'}$.
      \end{procedure}
    \end{center}
    \caption{$sort(\bott T)$}
  \end{subfigure}

  \caption{The two mutually recursive functions for computing $\Cc_f$.}
  \label{SortAndCollect}
\end{figure}

\begin{expl}
  The tree in \Cref{fig:tildeExTree} has the code $(2, 0, 2, 0, 0)$;
  its states are traversed in the following order: $A, B, C, D, E$ in
  the topmost call to $collect$.
\end{expl}

\begin{lem}\label{CodePrefixFree}
  The code $\Cc_f$ is prefix-free. That is, if $\bott T, \bott{T'}\in
  \FF_T$ are such that $\Cc_f(\bott T)$ is a prefix of $\Cc_f(\bott{T'})$,
  we have $\Cc_f(\bott T)=\Cc_f(\bott{T'})$.
\end{lem}
\begin{proof}
  Given a code $c$, and an index $i$, write $\delta(c,i)=\sum_{j=1}^i
  c_j - i$. This is the number of vertices that have been announced as
  children of vertices in $c_1, \dots, c_i$ but which are not themselves
  in $c_1, \dots, c_i$. Thus, if we are reading a code $c$, and we have
  read the $i$ first elements, we know that we must read at least
  $\delta(c,i)$ other elements. Moreover, remark that if
  $\delta(c,i)=0$, then we are at the end of the code, since we have
  already read the children of every vertex.

  Now, suppose that $c:=\Cc_f(\bott{T})$ is a prefix of
  $c':=\Cc_f(\bott{T'})$. Let $i=|c|$: we have $\delta(c,i)=0$ since
  $c$ is completely read once we have read the $i$ first elements.
  Moreover, we must have $\delta(c',i)=\delta(c,i)$ since
  $c'_1 \dots c'_i = c_1 \dots c_i$. So, $\delta(c',i)=0$ too, and thus, $c=c'$.
\end{proof}

We now say that, for two trees $\bott T,\bott{T'}$, we have $\bott
T\leq_f \bott{T'}$ if $\Cc_f(\bott{T})\leq_\text{lex} \Cc_f(\bott{T'})$.
We claim that this defines a total order on trees. Reflexivity and
transitivity are trivial, and its antisymmetry is guaranteed by the
following lemma:

\begin{lem}\label{OrderAntisymmetry}
  For any two finite trees $\bott T,\bott{T'}$, $\Cc_f(\bott
  T)=\Cc_f(\bott{T'})\implies \bott T=\bott{T'}$.
\end{lem}
\begin{proof}
  We just show that we can reconstruct $\bott T$ from $\Cc_f(\bott
  T)=collect(\bott T)$. We can ignore the call to $sort(\bott T)$ in
  $collect(\bott T)$: we can consider that the tree $\bott T$ we will
  recover is already sorted. To shorten notations, let's write
  $c:=\Cc_f(\bott T)$.

  First, we can partition $c$ into levels. Indeed, remark that if we
  know that the indices corresponding to states at depth $d$ form the
  set $\iitv{k,\ell}$, then we know that the number of states at depth
  $d+1$ is $\sum_{j=k}^\ell c_j$, and so the states at depth $d+1$
  correspond to indices $\iitv{\ell+1, \ell+\sum_{j=k}^\ell c_j}$. So,
  we can now iterate on the levels of $c$: let's write for convenience
  $k_d,\ell_d\in\N$ the first and last indices of states at depth $d$.

  The first level, corresponding to depth $0$, is easy to reconstruct:
  simply create the root. For our induction, we also create the
  predecessors of the root, of which we know the number, so the
  induction begins at depth $1$.

  Now, suppose we have uniquely reconstructed $\bott T$ up to depth
  $d$, and that we want to reconstruct level $d+1$. We traverse our
  reconstructed depth $1$ from left to right, and simultaneously
  traverse $c_{k_d}, \dots, c_{\ell_d}$. The $j$-th state we encounter
  has degree $c_{k_d+j}$, so we create $c_{k_d+j}$ children for that
  state. Thus, the level at depth $d+1$ is reconstructed uniquely too.

  Thus, we reconstruct $\bott T$, and this concludes the proof.
\end{proof}

The results which make this code useful are the following lemma and its
corollary.

\begin{lem}\label{TreeOrderProduct}
  For all finite trees $\bott T_1, \bott T_2, \bott T_3\in\FF_T$, we
  have $[\bott T_1]_{\depth(\bott T_3)}<_f[\bott T_2]_{\depth(\bott T_3)}\implies \bott T_1\bott T_3<_f\bott
  T_2\bott T_3$.
\end{lem}
\begin{proof}
%   In this proof, given a tree $\bott T$ and a state $s\in S_{\bott
%     T}$, we denote ${\Cc_f}_{\bott T}(s)$ the code of the tree anchored
%   on $s$.
  
  Let's prove this by induction on $\bott T_1$ and $\bott T_2$'s depth. It's
  trivial at depth $0$. Take $\bott T_1,\bott T_2$ of depth $\leq
  k+1$, with $k$ such that the result stands for trees of depth $\leq
  k$.

  Because of \Cref{CodePrefixFree}, since $\bott T_1<_f\bott T_2$,
  $\Cc_f(\bott T_1)$ cannot be a prefix of $\Cc_f(\bott T_2)$.

  Thus, there exists an index $i$ such that $\Cc_f(\bott
  T_1)_i<\Cc_f(\bott T_2)_i$ and for all $j<i$, we have $\Cc_f(\bott
  T_1)_j=\Cc_f(\bott T_2)_j$. Let $x\in S_{\bott T_1}$
  be the vertex at index $i$ in $\Cc_f(\bott T_1)$, and let $y\in
  S_{\bott T_2}$ be the vertex at index $i$ in $\Cc_f(\bott T_2)$. In
  the following, for a tree $\bott T$ and a vertex $u\in S_{\bott T}$,
  we denote by ${\Cc_f}_{\bott T}(u)$ the code of the subtree with
  root $u$ in $\bott T$. Since the codes share the same prefix of
  length $i-1$, $\depth_{\bott T_1}(x)=\depth_{\bott T_2}(y)$ (as seen
  in the proof of \Cref{OrderAntisymmetry}, this shared prefix of
  length $i-1$ holds all the information necessary to reconstruct
  everything above $x$ and $y$). Let's denote $d$ this depth. Because $[\bott T_1]_{\depth(\bott T_3)}<_f[\bott T_2]_{\depth(\bott T_3)}$, we have $d<\depth(\bott T_3)$.

  It is clear that we have $[\bott T_1]_{d-1}=[\bott T_2]_{d-1}$, so
  in particular, $[\bott T_1\bott T_3]_{d-1}=[\bott T_2\bott
    T_3]_{d-1}$. Let $z$ be the root of the tree with minimal code in $\bott T_3$ at depth $d$. Now, we show that the first difference between the
  codes of $\bott T_1\bott T_3$ and $\bott T_2\bott T_3$ is at the
  index $j$ corresponding to the vertex $(x,z)$ in $\Cc_f(\bott T_1\bott
  T_3)$, and to the vertex $(y,z)$ in $\Cc_f(\bott T_2\bott T_3)$. There
  might be multiple possibilities for $z$; we can assume that we take
  the one which gives the minimum $j$.

  Indeed, assume that a vertex of the form $(x,t)$ for some vertex $t$
  of $\bott T_3$ appears in $\Cc_f(\bott T_1\bott T_3)$ at depth $d$
  before index $j$. Since it appears before vertex $(x,z)$, by
  induction hypothesis, it means that the code of the subtree anchored
  on $t$ must be smaller than that of the subtree anchored on $z$. By
  minimality of $z$, this means that ${\Cc_f}_{\bott
    T_3}(z)={\Cc_f}_{\bott T_3}(t)$. Since we have chosen $z$ to be
  the first occurence of this code at this depth, we must have $t=z$.
  So, $(x,z)$ is the first vertex in $\Cc_f(\bott T_1\bott T_3)$ in
  which $x$ appears. A similar reasoning shows that no vertex
  involving $y$ appears before index $j$ in $\Cc_f(\bott T_2\bott
  T_3)$. Since every element before $x$ is shared between $\Cc_f(\bott
  T_1)$ and $\Cc_f(\bott T_2)$, this means that the first difference
  between $\Cc_f(\bott T_1\bott T_3)$ and $\Cc_f(\bott T_2\bott T_3)$
  is at or after index $j$.

  At index $j$, the number of predecessors $\Cc_f(\bott T_1\bott
  T_3)_j$ is $\npreds_{\bott T_1}(x)\npreds_{\bott T_3}(z)$ while
  $\Cc_f(\bott T_1\bott T_3)_j$ is $\npreds_{\bott
    T_2}(y)\npreds_{\bott T_3}(z)$. Since $\npreds_{\bott
    T_1}(x)<\npreds_{\bott T_2}(y)$, this shows that $\bott T_1\bott
  T_3<_f\bott T_2\bott T_3$.
\end{proof}

\begin{cor}\label{CompleteCodeOrder}
  For all finite trees $\bott T_1, \bott T_2, \bott T_3,\bott
  T_4\in\FF_T$, if $[\bott T_1]_{\depth(\bott T_3)}<_f[\bott T_2]_{\depth(\bott T_3)}$ and $[\bott T_3]_{\depth(\bott T_2)}\leq_f [\bott
  T_4]_{\depth(\bott T_2)}$, we have $\bott T_1\bott T_3<_f\bott T_2\bott T_4$.
\end{cor}
\begin{proof}
  By \Cref{TreeOrderProduct}, we have $\bott T_1\bott T_3<_f\bott
  T_2\bott T_3$. If $\bott T_3=\bott T_4$, we can conclude now.
  Otherwise, $\bott T_3<_f\bott T_4$, and we have, by
  \Cref{TreeOrderProduct}, $\bott T_2\bott T_3<_f\bott T_2\bott T_4$. Combining the two inequalities, we
  get: $\bott T_1\bott T_3<_f \bott T_2\bott T_4$.
\end{proof}

We are now ready for the recovery algorithm on finite trees.

\begin{lem}\label{FiniteTreeDiv}
  If $\bott A,\bott B,\bott C\in\FF_T$, $\bott A$ is finite, and
  $\bott A\bott B=\bott A\bott C$, then $[\bott B]_{\depth(\bott
    A)}=[\bott C]_{\depth(\bott A)}$.
\end{lem}
\begin{proof}
  Since $<_f$ is a complete order, if $[\bott B]_{\depth(\bott
    A)}\neq[\bott C]_{\depth(\bott A)}$, we can assume without loss of generality that we are in the case $[\bott B]_{\depth(\bott A)}<[\bott C]_{\depth(\bott A)}$. In that case, by \Cref{TreeOrderProduct}, we have $\bott A\bott B<_f\bott A\bott C$. This concludes.
\end{proof}

\begin{lem}\label{InfTreeDiv}
  If $\bott A,\bott B, \bott C\in \FF_T$ and $\bott A$ is infinite,
  and $\bott A\bott B=\bott A\bott C$, then $\bott B=\bott C$.
\end{lem}
\begin{proof}
  For every $d\in\N$, we have $[\bott A]_d[\bott B]_d=[\bott
    A]_d[\bott C]_d$, and thus, from \Cref{FiniteTreeDiv}, $[\bott
    B]_d=[\bott C]_d$. This implies that $\bott B=\bott C$.
\end{proof}

\begin{cor}\label{UniqTreeDiv}
  If $A,B,C\in\DD_D$, and $AB=AC$, then $B=C$.
\end{cor}
\begin{proof}
  If $AB=AC$, then $\tildeop A\tildeop B=\tildeop A\tildeop C$. Using
  \Cref{InfTreeDiv}, this means that $\tildeop B=\tildeop C$. Thus,
  $B=C$.
\end{proof}

We can now extend the order on possibly infinite trees; this will be of use for our results on the unicity of $k$-th roots. For a tree
$\bott T$, define its code as $\Cc(\bott T):=(\Cc_f([\bott
  T]_i))_{i\in \N}$, and say that $\bott T\leq\bott U$ if and only if
$\Cc(\bott T)\leq_{\text{lex}} \Cc(\bott U)$.

\begin{lem}\label{GeneralRestrictedProdOrder}
  For all trees $\bott T_1, \bott T_2, \bott T_3\in\FF_T$, if $\bott
  T_1<\bott T_2$, we have $\bott T_1\bott T_3<\bott T_2\bott T_3$.
\end{lem}
\begin{proof}
  If $\bott T_1<\bott T_2$, then $\bott T_1\neq\bott T_2$. In
  particular, there is a minimal depth $d$ such that $[\bott
    T_1]_d\neq[\bott T_2]_d$. Since for every $i<d$, we have $[\bott T_1]_i=[\bott T_2]_i$, we have $\Cc(\bott T_1\bott T_3)_1, \dots, \Cc(\bott T_1\bott
  T_3)_{d-1}=\Cc(\bott T_2\bott T_3)_1, \dots, \Cc(\bott T_2\bott
  T_3)_{d-1}$.
  
  What is left to prove is that $\Cc(\bott T_1\bott T_2)_d<\Cc(\bott T_1\bott T_3)_d$, that is $\Cc_f([\bott T_1\bott T_3]_d)<\Cc_f([\bott T_1\bott T_3]_d)$. This follows from the fact that $\Cc_f([\bott T_1]_d)<\Cc_f([\bott T_2]_d)$ and \Cref{TreeOrderProduct}.
\end{proof}

\begin{cor}\label{GeneralProdOrder}
  For all trees $\bott T_1, \bott T_2, \bott T_3,\bott T_4\in\FF_T$,
  if $\bott T_1<\bott T_2$ and $\bott T_3\leq \bott T_4$, we have
  $\bott T_1\bott T_3<\bott T_2\bott T_4$.
\end{cor}
\begin{proof}
    By \Cref{GeneralRestrictedProdOrder}, we have $\bott T_1\bott
    T_3<\bott T_2\bott T_3$. If $\bott T_3=\bott T_4$, we can conclude
    now. Otherwise, $\bott T_3<\bott T_4$, and we have, by
    \Cref{GeneralRestrictedProdOrder}, $\bott T_1\bott T_3<\bott
    T_2\bott T_3$ and $\bott T_2\bott T_3<\bott T_2\bott T_4$.
    Combining the two, we get: $\bott T_1\bott T_3< \bott T_2\bott
    T_4$.
\end{proof}

\subsection{Transforming an FDS into a forest}

In this subsection, we introduce a way of converting a general FDS into a
forest, since the product on forests works level by level. We do as follows:

\begin{defn}
    Let $A = C_n(\bott T_1, \dots, \bott T_n)$ be a connected FDS. For any $a\in S_A$, we write $A^{-k}(a):=\{s\in S_A:A^k(s)=a\}$. Then, for each $a\in [A]_0$, we set $$S_a:=\{(s,k):s\in A^{-k}(a),k\in\N\}$$
    and
    $$E_a:=\{((s,k),(A(s),k-1)):(s,k)\in S_a\}.$$
%    which define respectively the set of vertices and set of arcs of a tree we denote $\bott T_a(A)$.
\end{defn}

\begin{lem} \label{LemTa}
  The directed graph $\bott T_a(A)$ with vertex set $S_a$ and edge set $E_a$ defined above is a tree. Moreover $S_a$ and $S_b$ are disjoint for all $a \ne b$.
\end{lem}

\begin{proof}
%  We have to show that the $\bott T_a(A)$ are indeed well-defined trees, and we want to show that for $a\neq b$, $S_a\cap S_b=\es$ so that we can take the vertex set $S_{\tildeop A}$ to be $\bigcup_{a\in S_{[A]_0}}S_a$ without any renaming.
  
  Take $a\in[A]_0$. We will show that $\bott T_a(A)$ is a tree of root $(a,0)$. First, $\bott T_a(A)$ is acyclic because $k$ necessarily decreases following any arc, which also shows that $\bott T_a(A)$ is correctly oriented. Furthermore, if $b\in S_a$, then there exists $k$ such that $A^k(b)=a$, and thus we have the following path from $a$ to $b$:
  $$
    (b,k) \rightarrow (A(b), k-1) \rightarrow \dots \rightarrow (A^{k-1}(b), 1) \rightarrow (a,0).
  $$
  
  which has all of its edges in $E_a$. So, $T_a(A)$ is a well-defined tree, and $\tildeop{A}$ is a forest.
  
  Now, we show that if $a,b\in S_{[A]_0}$ and $a\neq b$, then $S_a\cap S_b=\es$. Suppose that $(s,k)\in S_a\cap S_b$. Then, $A^k(s)=a=b$, which is the desired contradiction. Thus, we may write $S_{\tildeop A}=\bigcup_{a\in S_{[A]_0}}S_a$ without renaming.
\end{proof}

We thus define the \textit{unrolling} of $A$ as $\tildeop{A}:=\sum_{a\in S_{[A]_0}} \bott T_a(A)$, with $S_{\tildeop A}=\bigcup_{a\in S_{[A]_0}}S_a$.

We can then extend this to general FDSs, by writing:
$\tildeop{A+B}=\tildeop{A}+\tildeop{B}$. Note that the unrolling is not injective. Indeed, for
instance, $\tildeop{C_3}=\tildeop{3C_1}$. This is not true even for FDSs with the same periodic part: if $\bott T$ and $\bott U$ are two distinct trees and $X = 2 C_1( \bott T ) + C_2( \bott U, \bott U )$ and $Y = 2 C_1( \bott U ) + C_2( \bott T, \bott T )$, then $\tildeop{ X } = \tildeop{ Y }$. However, in the connected case, we have injectivity.

\begin{lem} \label{UnrollingInjective}
Let $X,Y \in \DD$. If $X$ and $Y$ are connected and $[X]_0 = [Y]_0$, then $\tildeop{ X } = \tildeop{ Y } \implies X = Y$. 
\end{lem}

\begin{proof}
Let $X = C_x(\bott T_1, \dots, \bott T_x)$. Then $\tildeop{ X }$ has $x$ infinite trees $\bott X_1, \dots, \bott X_x$, each a periodic shift of the previous one: 
\[
    tseq(\bott X_1) = (\bott T_1, \bott T_2, \dots, \bott T_x), \dots, tseq(\bott X_x) = (\bott T_x, \bott T_1, \dots, \bott T_{x-1}).
\]
We have $Y = C_x( \bott U_1, \dots, \bott U_x )$, and similarly $\tildeop{ Y }$ consists of the trees $\bott Y_1, \dots, \bott Y_x$ where
\[
    tseq(\bott Y_1) = (\bott U_1, \bott U_2, \dots, \bott U_x), \dots, tseq(\bott U_x) = (\bott U_x, \bott U_1, \dots, \bott U_{x-1}).
\]
Then $\bott U_1 \in \{ \bott T_1, \dots, \bott T_x \}$, without loss say $\bott U_1 =\bott  T_1$, then $\bott U_y = \bott T_y$ for all $1 \le y \le x$ and $X = Y$.
\end{proof}

\begin{expl}
  See \Cref{fig:tildeExTrees} and \Cref{fig:tildeExConn}.

  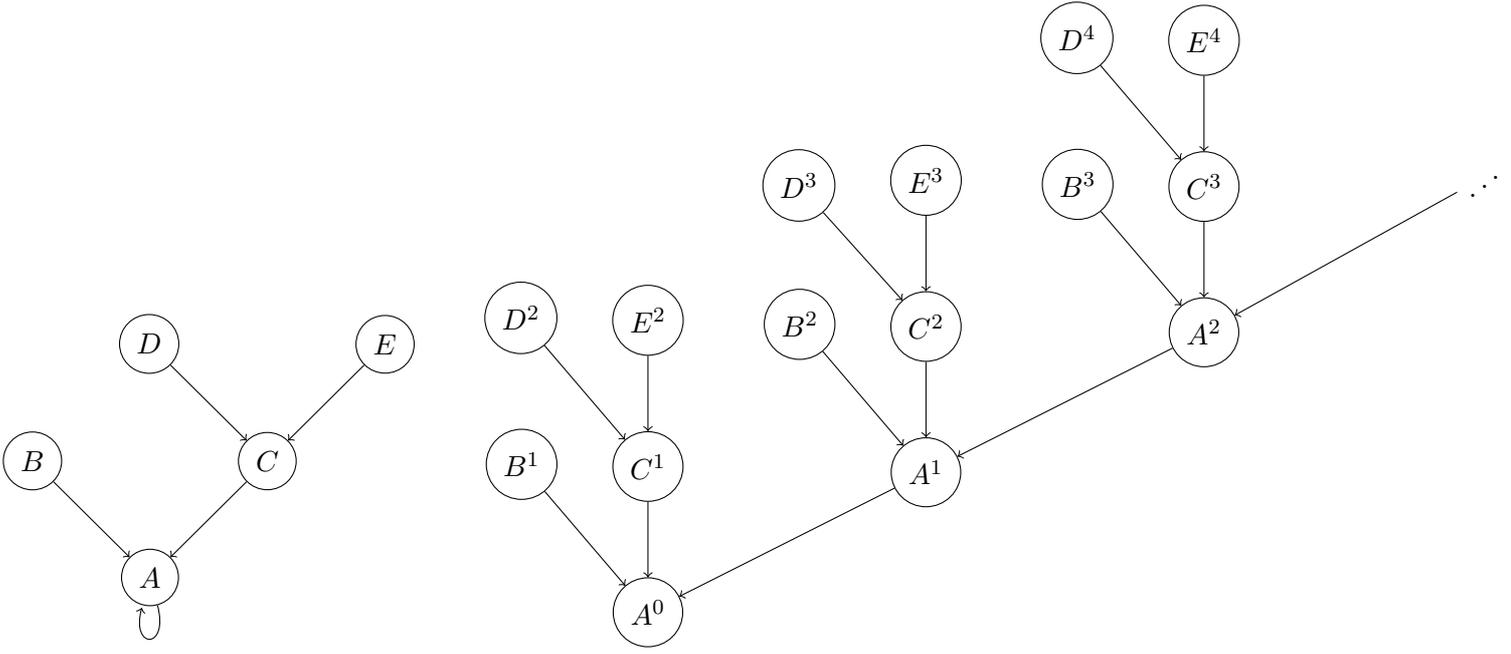
\begin{figure}
    \hspace{-2.4cm}\noindent\makebox[\textwidth]{%
    \begin{subfigure}[B]{0.4\textwidth}
      \begin{center}
        \begin{tikzpicture}
          \node[draw,circle] (A) {$A$};
          \node[draw,circle] (B) [above left=of A] {$B$};
          \node[draw,circle] (C) [above right=of A] {$C$};
          \node[draw,circle] (D) [above left=of C] {$D$};
          \node[draw,circle] (E) [above right=of C] {$E$};
          \draw[->] (D) -- (C);
          \draw[->] (C) -- (A);
          \draw[->] (B) -- (A);
          \draw[->] (E) -- (C);
          \draw[->] (A) edge [loop below] (A);
        \end{tikzpicture}
      \end{center}
      \caption{A dendron $T$.}
      \label{fig:tildeExTree}
    \end{subfigure}
    \begin{subfigure}[B]{0.5\textwidth}
      \begin{center}
        \begin{tikzpicture}
          \node[draw,circle] (A0) {$A^0$};
          \node[draw,circle] (B) [above left=1.3cm and 1cm of A0] {$B^1$};
          \node[draw,circle] (C) [above=of A0] {$C^1$};
          \node[draw,circle] (D) [above left=1.3cm and 1cm of C] {$D^2$};
          \node[draw,circle] (E) [above=of C] {$E^2$};
          \draw[->] (D) -- (C);
          \draw[->] (C) -- (A0);
          \draw[->] (B) -- (A0);
          \draw[->] (E) -- (C);

          \node[draw,circle] (A1) [above right=1.2cm and 3cm of A0] {$A^1$};
          \node[draw,circle] (B) [above left=1.3cm and 1cm of A1] {$B^2$};
          \node[draw,circle] (C) [above=of A1] {$C^2$};
          \node[draw,circle] (D) [above left=1.2cm and 1cm of C] {$D^3$};
          \node[draw,circle] (E) [above=of C] {$E^3$};
          \draw[->] (D) -- (C);
          \draw[->] (C) -- (A1);
          \draw[->] (B) -- (A1);
          \draw[->] (E) -- (C);
          \draw[->] (A1) -- (A0);

          \node[draw,circle] (A2) [above right=1.2cm and 3cm of A1] {$A^2$};
          \node[draw,circle] (B) [above left=1.3cm and 1cm of A2] {$B^3$};
          \node[draw,circle] (C) [above=of A2] {$C^3$};
          \node[draw,circle] (D) [above left=1.3cm and 1cm of C] {$D^4$};
          \node[draw,circle] (E) [above=of C] {$E^4$};
          \draw[->] (D) -- (C);
          \draw[->] (C) -- (A2);
          \draw[->] (B) -- (A2);
          \draw[->] (E) -- (C);
          \draw[->] (A2) -- (A1);

          \node (etc) [above right=1.3cm and 3cm of A2] {$\iddots$};
          \draw[->] (etc) -- (A2);
        \end{tikzpicture}
      \end{center}
      \caption{The tree $\tildeop{T}$, rooted in $A^0$.}
    \end{subfigure}
    }
    \caption{The $\tildeop{\cdot}$ operation on a dendron.}
    \label{fig:tildeExTrees}
  \end{figure}

  \begin{figure}
    \begin{subfigure}[B]{0.3\textwidth}
      \begin{center}
        \begin{tikzpicture}
          \node[draw,circle] (A) {$A$};
          \node[draw,circle] (F) [below=of A] {$F$};
          \node[draw,circle] (G) [below=of F] {$G$};
          \node[draw,circle] (B) [above left=of A] {$B$};
          \node[draw,circle] (C) [above right=of A] {$C$};
          \node[draw,circle] (D) [above left=of C] {$D$};
          \node[draw,circle] (E) [above right=of C] {$E$};
          \draw[->] (D) -- (C);
          \draw[->] (C) -- (A);
          \draw[->] (B) -- (A);
          \draw[->] (E) -- (C);
          \draw[->] (A) to[bend left] (F);
          \draw[->] (F) to[bend left] (A);
          \draw[->] (G) -- (F);
        \end{tikzpicture}
      \end{center}
      \caption{A connected FDS $S$.}
    \end{subfigure}
    \begin{subfigure}[B]{0.7\textwidth}
      \begin{center}
        \begin{tikzpicture}
          \node[draw,circle] (A0) {$A^0$};
          \node[draw,circle] (B) [above left=1.2cm and 1.5cm of A0] {$B^1$};
          \node[draw,circle] (C) [above=of A0] {$C^1$};
          \node[draw,circle] (D) [above left=1.2cm and 1.5cm of C] {$D^2$};
          \node[draw,circle] (E) [above=of C] {$E^2$};
          \draw[->] (D) -- (C);
          \draw[->] (C) -- (A0);
          \draw[->] (B) -- (A0);
          \draw[->] (E) -- (C);

          \node[draw,circle] (F1) [above right=1.2cm and 1.5cm of A0] {$F^1$};
          \node[draw,circle] (G2) [above=of F1] {$G^2$};
          \draw[->] (G2) -- (F1);
          \draw[->] (F1) -- (A0);

          %% \node[draw,circle] (A1) [above right=of A0] {$A^1$};
          %% \node[draw,circle] (B) [above left=0.5cm of A1] {$B^2$};
          %% \node[draw,circle] (C) [above=of A1] {$C^2$};
          %% \node[draw,circle] (D) [above left=of C] {$D^3$};
          %% \node[draw,circle] (E) [above=of C] {$E^3$};
          %% \draw[->] (D) -- (C);
          %% \draw[->] (C) -- (A1);
          %% \draw[->] (B) -- (A1);
          %% \draw[->] (E) -- (C);
          %% \draw[->] (A1) -- (A0);

          \node[draw,circle] (A2) [above right=1.2cm and 1.5cm of F1] {$A^2$};
          \node[draw,circle] (B) [above left=1.2cm and 1cm of A2] {$B^3$};
          \node[draw,circle] (C) [above=of A2] {$C^3$};
          \node[draw,circle] (D) [above left=of C] {$D^4$};
          \node[draw,circle] (E) [above=of C] {$E^4$};
          \draw[->] (D) -- (C);
          \draw[->] (C) -- (A2);
          \draw[->] (B) -- (A2);
          \draw[->] (E) -- (C);
          \draw[->] (A2) -- (F1);

          \node (etc) [above right=of A2] {$\iddots$};
          \draw[->] (etc) -- (A2);

          \tikzxshift{7cm}

          \node[draw,circle] (F0) {$F^0$};
          \node[draw,circle] (G1) [above=of F0] {$G^1$};
          \draw[->] (G1) -- (F0);

          %% \node[draw,circle] (A1) [above right=of A0] {$A^1$};
          %% \node[draw,circle] (B) [above left=0.5cm of A1] {$B^2$};
          %% \node[draw,circle] (C) [above=of A1] {$C^2$};
          %% \node[draw,circle] (D) [above left=of C] {$D^3$};
          %% \node[draw,circle] (E) [above=of C] {$E^3$};
          %% \draw[->] (D) -- (C);
          %% \draw[->] (C) -- (A1);
          %% \draw[->] (B) -- (A1);
          %% \draw[->] (E) -- (C);
          %% \draw[->] (A1) -- (A0);

          \node[draw,circle] (A1p) [above right=1.2cm and 1.5cm of F0] {$A^1$};
          \node[draw,circle] (Bp) [above left=1.2cm and 1cm of A1p] {$B^2$};
          \node[draw,circle] (Cp) [above=of A1p] {$C^2$};
          \node[draw,circle] (Dp) [above left=1.2cm and 1cm of Cp] {$D^3$};
          \node[draw,circle] (Ep) [above=of Cp] {$E^3$};
          \node (etcp) [above right=of A1p] {$\iddots$};
          \draw[->] (Dp) -- (Cp);
          \draw[->] (Cp) -- (A1p);
          \draw[->] (Bp) -- (A1p);
          \draw[->] (Ep) -- (Cp);
          \draw[->] (A1p) -- (F0);
          \draw[->] (etcp) -- (A1p);
        \end{tikzpicture}
      \end{center}
      \caption{The forest $\tildeop{S}$, with roots $A^0$ and $F^0$.}
    \end{subfigure}
    \caption{The $\tildeop{\cdot}$ operation on a connected FDS.}
    \label{fig:tildeExConn}
  \end{figure}
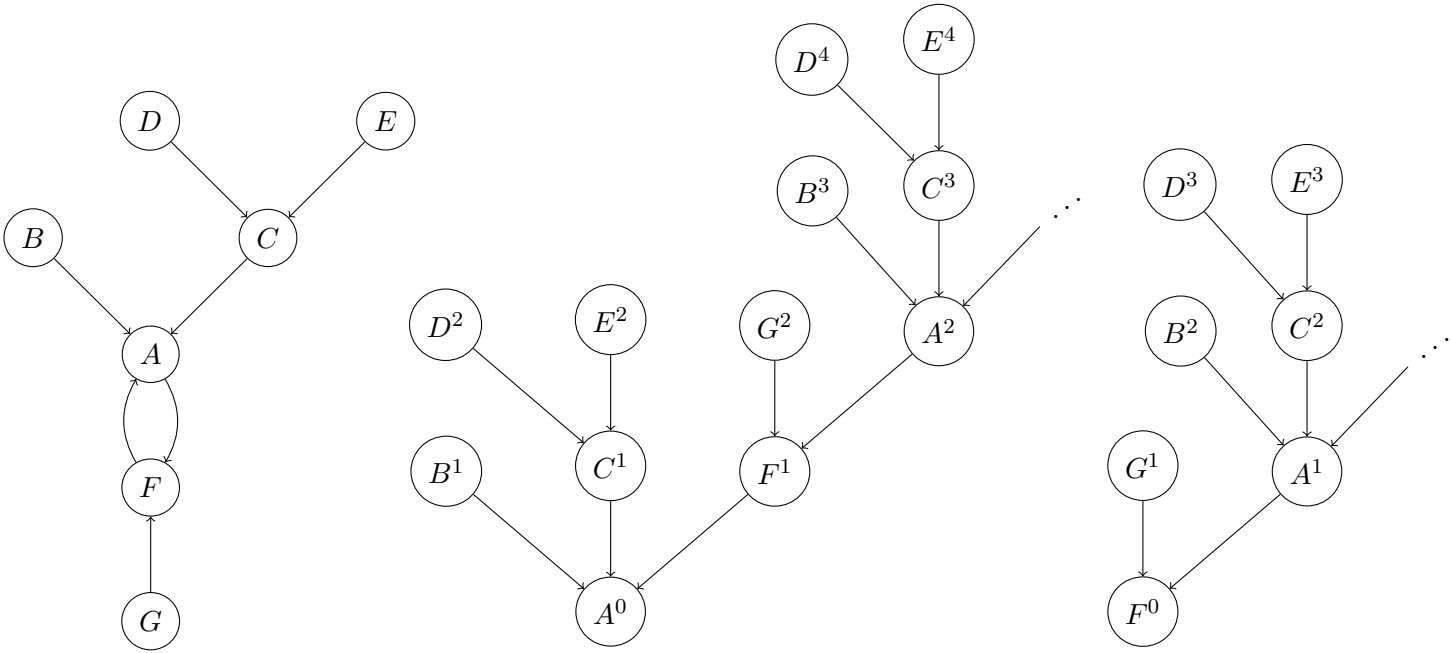
\end{expl}

The following lemma explains why the unrolling operation makes sense: it
is compatible with the product. The proof is rather technical, but the
intuition for this result is simple. A cycle behaves very much like an
infinite path in terms of predecessors,and the unrolling converts the cycle into an infinite path
that behaves similarly. Moreover, the reason we create multiple
infinite trees for each cycle is to avoid problems with cases where
the product of two connected FDSs gives a non-connected FDS.

\begin{lem}\label{ProdBot}
  For any $A,B\in\DD$, we have: $\tildeop{A}\tildeop{B}=\tildeop{AB}$.
\end{lem}
\begin{proof}
We show this result for connected $A$ and $B$ as the the other cases follow by distributivity. Thus, we write $A=C_m(\bott T_0,\dots,\bott T_{m-1})$ and $B=C_n(\bott U_0, \dots, \bott U_{n-1})$. Now, we can write:

\begin{eqnarray*}
    S_{\tildeop A}&=&\bigcup_{a\in[A]_0}\{(s,k):s\in A^{-k}(a),k\in\N\} \\
    S_{\tildeop B}&=&\bigcup_{b\in[B]_0}\{(s,k):s\in B^{-k}(b),k\in\N\} \\
\end{eqnarray*}

Now, the product $\tildeop A\tildeop B$ has the following state set:

\begin{eqnarray*}
    S_{\tildeop A\tildeop B}&=&\{(a,b)\in S_{\tildeop A}\times S_{\tildeop B}:\depth_{\tildeop A}(a)=\depth_{\tildeop B}(b)\} \\
    &\simeq& \{((s_a,k_a),(s_b,k_b))\in S_{\tildeop A}\times S_{\tildeop B}:k_a=k_b\} \\
    &\simeq& \bigcup_{(a,b)\in[A]_0\times [B]_0}\{(s_a, s_b, k)\in S_{A}\times S_{B}\times \N: s_a\in A^{-k}(a), s_b\in B^{-k}(b);k\in\N\}.
\end{eqnarray*}

Now, let's show that this is isomorphic to $S_{\tildeop{AB}}$ (remember that $S_{AB}=S_A\times S_B$):

\begin{eqnarray*}
    S_{\tildeop{AB}} &=& \bigcup_{c\in [AB]_0}\{(s,k)\in S_{AB}\times\N:s\in AB^{-k}(c),k\in\N\} \\
    % &=& \bigcup_{(a,b)\in[A]_0\times [B]_0}\{(s,k)\in S_{AB}\times\N:s\in AB^{-k}((a,b)),k\in\N\} \\
    &=& \bigcup_{(a,b)\in[A]_0\times [B]_0}\{((s_a,s_b),k)\in S_{AB}\times\N:(s_a,s_b)\in AB^{-k}((a,b)),k\in\N\} \\
    &=& \bigcup_{(a,b)\in[A]_0\times [B]_0}\{((s_a,s_b),k)\in S_{AB}\times\N:s_a\in A^{-k}(a), s_b\in B^{-k}(b),k\in\N\}.
\end{eqnarray*}

The last step comes from the following identity: for $c=(a,b)\in S_{AB}=S_A\times S_B$, we have $AB^{-k}(c) = A^{-k}(a)\times B^{-k}(b)$. Thus, we have shown that $S_{\tildeop A\tildeop B} \simeq S_{\tildeop{AB}}$.

Now, what is left to do is show that the edges are also isomorphic. Thus, we must show that for any $(s_a,s_b,k), (s'_a,s'_b,k')\in S_{\tildeop A\tildeop B}$, we have $(s_a,s_b,k)\rightarrow (s'_a,s'_b,k')$ in $\tildeop A\tildeop B$ if and only if we have $(s_a,s_b,k)\rightarrow (s'_a,s'_b,k')$ in $\tildeop{AB}$.

In the end of the proof, we denote $x\rightarrow{C} y$ the existence of an edge from $x$ to $y$ in the forest or FDS $C$ (if $C$ is an FDS, $x\rightarrow{C} y$ means $y=C(x)$). Now, we can reason by equivalence:

\begin{eqnarray*}
    && (s_a,s_b,k)\xrightarrow{\tildeop A\tildeop B} (s'_a,s'_b,k') \\
    &\iff& k'=k-1 \wedge (s_a,k)\xrightarrow{\tildeop A} (s'_a,k') \wedge (s_b,k)\xrightarrow{\tildeop B} (s'_b,k') \\
    &\iff& k'=k-1 \wedge s_a\xrightarrow{A} s'_a \wedge s_b\xrightarrow{B} s'_b \\
    &\iff& k'=k-1 \wedge (s_a,s_b)\xrightarrow{AB} (s'_a,s'_b) \\
    &\iff& (s_a,s_b,k)\xrightarrow{\tildeop{AB}} (s'_a,s'_b,k').
\end{eqnarray*}

This concludes.
\end{proof}

We can now show that division is unambiguous when restricted to connected FDSs.

\begin{thm} \label{thmConnected}
For any FDS $A \in \DD$, if $X, Y \in \DD$ are connected, then
\[
    AX = AY \implies X = Y.
\]
\end{thm}

\begin{proof}
Suppose $AX = AY$. Let $[X]_0 = C_x$ and $[Y]_0 = C_y$, then $|[AX]_0| = x |[A]_0|$ and $|[AY]_0| = y |[A]_0|$ show that $x = y$, that is $[X]_0 = [Y]_0$. Thus,
\[
    AX = AY 
    \implies \tildeop{AX} = \tildeop{AY} 
    \xRightarrow{ \text{ \Cref{ProdBot} } }  \tildeop{A} \tildeop{X} = \tildeop{A} \tildeop{Y} 
    \xRightarrow{ \text{ \Cref{InfTreeDiv} } } \tildeop{X} = \tildeop{Y} 
    \xRightarrow{ \text{ \Cref{UnrollingInjective} } } X = Y.
\]
\end{proof}

We remark that \Cref{thmConnected} implies \cite[Conjecture 3.1]{DFPR22}. Indeed, if $A$ and $B$ are connected and $AX = AY = B$, then $X$ and $Y$ are connected, thus $X = Y$.

\subsection{Cancellative FDSs are those with a fixpoint}

Using the results of the previous part, we have the following lemma:

\begin{lem}\label{FixpointIsCancellative}
  If $A\in\DD$, and $A$ has a fixpoint, then $A$ is cancellative.
\end{lem}
\begin{proof}
  Take $B,D\in\DD$ such that $AB=D$. Let's show that we can recover
  $B$ by induction on the size of $D$. The base case is trivial: if
  $|D|=0$, then $D=0$ and since $A$ has a fixpoint, $A\neq 0$, so
  $B=0$.

  Denote $\ell$ the size of the smallest cycle of $D$. Since $A$ has a
  cycle of length $1$, it means that the smallest cycle of $B$ is of
  length $\ell$ too. Let $L \subseteq \N$ be the set of divisors of $\ell$. We denote $A' = \supp_L(A)$, and similarly $B' = \supp_L(B)$ and $D' = \supp_L(D)$.
%   Now, consider $A'$ the FDS made of the connected
%   components of $A$ which have cycle lengths that divide $\ell$.
%   Moreover, take $B'$ the FDS made of connected components of $B$
%   which have cycle of lengths of $\ell$, and $D'$ the analogue for
%   $D$. 
  Then we have $A'B'=D'$. Indeed, cycles of length $\ell$ in $D$
  come from a product of a cycle of length $a$ in $A$ and length $b$
  in $B$, such that $a\vee b=\ell$. In particular, this implies that
  $a|\ell$, and since $b\geq\ell$ because $\ell$ is the smallest cycle
  length in $B$, this implies $b=\ell$.

  So, we have $A'B'=D'$, which implies
  $\tildeop{A'}\tildeop{B'}=\tildeop{D'}$. Take the smallest tree in
  $\tildeop{A'}$, denote it $\bott T_A$, and take the smallest tree in
  $\tildeop{D'}$, denote it $\bott T_D$. Then, there is a tree $\bott
  T_B$ in $\tildeop{B'}$ such that $\bott T_A\bott T_B=\bott T_D$, by
  \Cref{CompleteCodeOrder} and minimality of $\bott T_A$ and $\bott
  T_D$.

  This means that by \Cref{InfTreeDiv}, we find $\bott T_B$ by
  dividing $\bott T_D$ by $\bott T_A$. Moreover, since $\bott T_B$ is
  in $\tildeop{B'}$, we know that it comes from a cycle of length
  $\ell$ in $B$. So, we set $E=C_\ell(tseq(\bott T_B)_1,\dots,tseq(\bott T_B)_\ell)$ the ``reconstruction'' of this cycle. The useful
  property of $E$ is that it is part of $B$. Thus, the equation
  becomes $A(B-E)=D-AE$ (those two subtractions are well-defined since $E$ is a connected component of $B$, and $AE$ is a connected component of $D$), which involves a product strictly smaller
  than $D$.
\end{proof}

Now, we show that if an FDS has no fixpoint, then it is not
cancellative.

\begin{lem}\label{LemChinese}
  Let $\Aa$ be a finite set of integers greater than $1$. Then there exist $X \ne X' \in \DD_P$ such that $C_a X = C_a X'$ for all $a \in \Aa$.
%   Then there exists
%   $(Y_a)_{a\in \Aa}\in\DD_P^\Aa$ such that the system $\forall a\in
%   \Aa, C_{a}X=Y_a$ has at least two solutions in $\DD_P$.
\end{lem}
\begin{proof}
    Recall the sequence $\delta_J$ from Lemma \ref{LemDelta}. For all $I \subseteq \Aa$, let $\alpha_I = \delta_\Aa\prod_{a\in \Aa} a$ and $\alpha'_I = \alpha_I+(-1)^{|I|}\delta_I\prod_{a\in A\setminus
    I}a$.
    
    Since $\alpha_I, \alpha'_I \ge 0$, we can then define the FDSs $X=\sum_{I\subs \Aa}\alpha_IC_{\bigvee I}$ and $X' = \sum_{I \subseteq \Aa} \alpha'_I C_{\bigjoin I}$. We remark that the number of fixpoints in $X$ and $X'$ are $\alpha_\es$ and $\alpha'_\es$, respectively. Since $\alpha'_\es = \alpha_\es + \prod_{a \in \Aa} a \ne \alpha_\es$, $X$ and $X'$ are distinct FDSs.

Let $b \in \Aa$. For all $I \subseteq \Aa \setminus \{b\}$, let $J = I \cup \{b\}$.
Then we have
  \begin{eqnarray*}
    C_{b}(\alpha'_IC_{\bigvee I} + \alpha'_{J}C_{\bigvee J}) &=&
    (\alpha'_I(b\wedge \bigvee I) + \alpha'_Jb)C_{\bigvee J} \\
    &=& ((\alpha_I+(-1)^{|I|}\delta_I\prod_{a\in \Aa \setminus I}a)(b\wedge \bigvee I) +
    (\alpha_J-(-1)^{|I|}\delta_J\prod_{a\in A\setminus J}a)b)C_{\bigvee J} \\
    &=& ( \alpha_I (b \meet \bigjoin I) + \alpha_J b ) C_{\bigjoin J} \\
    && + \left[ \left( (-1)^{|I|} \delta_I \prod_{a \in A \setminus I} a \right) (b \meet \bigjoin I) - \left( (-1)^{|I|} \delta_I  (b \meet \bigjoin I) \prod_{a \in A \setminus I} a  \right) \right] C_{\bigjoin J} \\
    &=& C_{b}(\alpha_IC_{\bigvee I} + \alpha_{J}C_{\bigvee J}).
  \end{eqnarray*}

Therefore,
\[
    C_b X' = \sum_{ I \subseteq \Aa \setminus \{b\} } C_{b}(\alpha'_IC_{\bigvee I} +
  \alpha'_{J}C_{\bigvee J}) \\
  =  \sum_{ I \subseteq \Aa \setminus \{b\} } C_{b}(\alpha_IC_{\bigvee I} + \alpha_{J}C_{\bigvee J}) \\
  = C_b X.
\]
\end{proof}

This lemma above combined with \Cref{FixpointIsCancellative} gives:

\begin{thm} \label{Cancellative}
An FDS is cancellative if and only if it has a fixpoint.
%  If $A\in\DD$, $A$ is cancellative if and only if it has a fixpoint.
\end{thm}

\begin{proof}
The case where the FDS has a fixpoint is handled by \Cref{FixpointIsCancellative}. Suppose $A$ has no fixpoint and let $\Aa$ be the set of all cycle lengths of $A$. Following \Cref{LemChinese}, there exist $X, X' \in \DD_P$ such that $C_a X = C_a X'$ for all $a \in \Aa$. Let $B = C_a( \bott T_1, \dots, \bott T_2  )$ be a connected component of $A$, where $a \in \Aa$. According to \Cref{CorMultPermTrees}, we have $BX = BX'$. Summing over all connected components of $A$, we finally obtain $AX = AX'$.
\end{proof}

From now on, we define $\DD^*$ to be the set of cancellable
FDSs. Its algebraic structure is that of a cancellative
subsemiring of $\DD$, but $\DD^*$ does not have an additive identity.

\section{Polynomial-time algorithm for tree and dendron division}\label{sec:Algos}

The algorithm \Cref{divideAlgo} provides an algorithmic proof of \Cref{FiniteTreeDiv}, as formalised below:

\begin{figure}
  \begin{center}
    \begin{procedure}[H]
      $\Mm_{\bott C}\leftarrow \Dd(\bott C)$\;
      $\Mm'\leftarrow\es$\;

      \While{$\Mm_{\bott C}\neq\es$}{
        $d\leftarrow \depth(\bott C)-1$\;
        $\Tt_{\bott C}\leftarrow \{\{\bott X\in\Mm_{\bott C}:\depth(\bott X)\geq d\}\}$\;
        $\Tt_{\bott A}\leftarrow \{\{\bott Y\in\Dd(\bott A):\depth(\bott Y)\geq d\}\}$\;
        $\bott t_{\bott C} \leftarrow \argmin_{\bott X\in\Tt_{\bott C}}
        \Cc_f(\bott X)$\;
        $\bott t_{\bott A} \leftarrow \argmin_{\bott Y\in\Tt_{\bott A}}
        \Cc_f(\bott Y)$\;
        $\bott t_{\bott B}\leftarrow divide(\bott t_{\bott C},\bott t_{\bott A})$\;
        \If{$\bott t_{\bott B}=\bot$ or $\bott t_{\bott B}\Dd(\bott A)\centernot\subseteq \Mm_{\bott C}$}{
          \Return $\bot$\;
        }
        $\Mm_{\bott C}\leftarrow \Mm_{\bott C} \setminus \bott t_{\bott B}\Dd(\bott
        A)$\;
        $\Mm'\leftarrow \Mm'\cup \{\bott t_{\bott B}\}$\;
      }
      Let $\bott B$ such that $\Dd(\bott B)=\Mm'$\;
      \Return $\bott B$\;
    \end{procedure}
  \end{center}
  \caption{$divide(\bott C,\bott A)$ to divide $\bott C$ by $\bott A$,
    for finite $\bott C$ and $\bott A$.}
  \label{divideAlgo}
\end{figure}

\begin{lem}
  The $divide$ algorithm is correct: for all $\bott A,\bott B, \bott C\in \FF_T$, $\bott{A}\bott{B}=\bott{C}\implies [\bott{B}]_{\depth(\bott A)} =
divide(\bott C, \bott A)$, and $[\bott C]_{\depth(\bott A)}\centernot|\bott A\implies divide(\bott C,\bott A)=\bot$.
\end{lem}
\begin{proof}
  In the case in which $\bott A\bott B=\bott C$, we show that we can recover uniquely $[\bott B]_{\depth(\bott A)}$ from $\bott A$ and $\bott A\bott B$ by induction on $\depth(\bott A)$.
  The base case is for $\depth(\bott A)=-1$, in which $\bott A=\bott
  0$ is the empty tree. Then, the result is trivial since $[\bott
    B]_{-1}=\bott 0$ for any $\bott B\in\FF_T$.

  Now, for the general case, we do an induction on the size of the
  product $\bott C=\bott A\bott B$. The base case for $\bott C= \bott 0$ is
  trivial. Let's write $\{\{\bott T_1, \dots, \bott T_n\}\}=\Dd(\bott
  A)$ with $\bott T_1\leq_f \dots\leq_f \bott T_n$, $\{\{\bott U_1, \dots,
  \bott U_k\}\}=\Dd(\bott B)$ with $\bott U_1\leq_f \dots\leq_f \bott
  U_k$, and finally, write $\{\{\bott V_1, \dots, \bott
  V_{nk}\}\}=\Dd(\bott C)$ with $\bott V_1\leq_f \dots\leq_f \bott
  V_{nk}$. We remark that to recover $[\bott B]_{\depth(\bott
    A)}$, all we need is to recover $[\bott U_j]_{\depth(\bott A) - 1}$ for all $1 \le j \le k$.

  Let $d$ be $\depth(\bott C)-1$ as in the algorithm. Then let $\bott t_{\bott A}$ (respectively $\bott t_{\bott B}$, $\bott t_{\bott C}$) be the minimum tree in $\Dd(\bott A)$ (respectively $\Dd(\bott B)$, $\Dd(\bott C)$) of depth $\geq d$.
%   Then, among the trees of
%   depth $\geq d-1$ in $\bott T_1, \dots, \bott T_n$, take the maximum
%   $\bott T_v$ (according to the order on trees), and similarly, take
%   the maximum subtree of depth $d-1$ among $\bott V_1, \dots, \bott
%   V_{nk}$, denoted $\bott V_w$. 
  We can then write $\bott t_{\bott A}\bott t_{\bott B}=\bott t_{\bott C}$ without loss of generality. Since $\bott t_{\bott A}$ has
  depth $<\depth(\bott A)$, the outer induction hypothesis shows that $divide(\bott t_{\bott C},\bott t_{\bott A})=[\bott t_{\bott B}]_d$.

  There are two cases. If $\depth(\bott B)\leq\depth(\bott A)$, then
  $d=\depth(\bott B)$ by \Cref{DepthTrees} and so $[\bott t_{\bott B}]_d=\bott t_{\bott B}$. Otherwise, if $\depth(\bott B)>\depth(\bott A)$, then $\depth(\bott C)=\depth(\bott A)$ by \Cref{DepthTrees}
  and so $\bott t_{\bott B}=[\bott t_{\bott B}]_{\depth(\bott A)-1}$, which is a depth $1$ subtree of $[\bott B]_{\depth(\bott A)}$. So, in both cases, $\bott t_{\bott B}$ is a depth $1$ subtree of $[\bott B]_{\depth(\bott A)}$.

  Now that we have $\bott t_{\bott B}$, the algorithm computes $\bott t_{\bott B}\Dd(\bott A) = \{\{\bott t_{\bott B}\bott T_1, \dots, \bott t_{\bott B}\bott T_n\}\}$, which are $n$ subtrees of $\bott C$, and removes them from $\bott C$. Finally, the next loop iteration corresponds to applying the internal induction hypothesis to the identity $\bott A\bott{B'}=\bott{D'}$ where 
    $$
        \Dd(\bott{B'}) = \Dd([\bott B]_{\depth(\bott A)}) \setminus \{\bott t_{\bott B}\}
    $$ 
    and
    $$
    \Dd(\bott{D'}) = \Dd(\bott{D})\setminus\bott t_{\bott B}\Dd(\bott A).
    $$
    
    To conclude, if we are in the case where $[\bott C]_{\depth(\bott A)}\centernot|\bott A$, we need to show that if $divide([\bott C]_{\depth(\bott A)},\bott A)$ does not return $\bot$ but some tree $\bott B$, then $\bott A\bott B=[\bott C]_{\depth(\bott A)}$ which is a contradiction. To do so, remark that by construction during the while loop, $\Dd(\bott A)\Dd(\bott B)=\Dd([\bott C]_{\depth(\bott A)})$, which means that $\bott A\bott B=[\bott C]_{\depth(\bott A)}$.
\end{proof}

This algorithm only works on trees. But \Cref{LemTruncUnrolling} allows one
to use it on dendrons, using the truncature of their unrollings. First, we need the following definition, adapting the definition of product isomorphism for forests:

\begin{defn}
  Given a product $\bott B = \prod_{i\in I} \bott A_i$ for some finite set $I$, a family $(\bott A_i)_{i\in I}\in \FF^I$, and denoting $S_{\prod_{i\in I} \bott A_i}=\bigcup_{k\in\N}\{(a_i)_{i\in I}\in\prod_{i\in I} S_{\bott A_i}:\depth_{\bott A_i}(a_i)=k\}$, we say that the function $\psi: S_{\bott B}\mapsto S_{\prod_{i\in I} \bott A_i}$ is a \textit{forest product isomorphism for the product} $\bott B = \prod_{i\in I} \bott A_i$ if:
  \begin{enumerate}
      \item it is a bijection,
      \item for any $b\in S_{\bott B}$, $\psi(b)$ is a root if and only if $b$ is a root, and
      \item for any families of non-root states $(s_i)_{i\in I},(s'_i)_{i\in I}\in S_{\prod_{i\in I}\bott A_i}$, we have: $\psi^{-1}((s_i)_{i\in I})\rightarrow \psi^{-1}((s'_i)_{i\in I})$ is an edge of $\bott B$ if and only if for each $i\in I$, $s_i\rightarrow s'_i$ is an edge of $\bott A_i$.
  \end{enumerate}
\end{defn}

As for the first definition of a product isomorphism, if there is a tree product isomorphism between $\bott B$ and $\prod_{i\in I} \bott A_i$, this means that $\bott B=\prod_{i\in I} \bott A_i$. A simple inductive proof shows that:

\begin{lem}\label{TreeProdIsoDepth}
    Given a tree product isomorphism $\psi$ for a product $\bott B=\prod_{i\in I} \bott A_i$ is such that for any $(a_i)_{i\in I}\in S_{\prod_{i\in I} \bott A_i}$ and $b\in S_{\bott B}$, such that $\psi(b)=(a_i)_{i\in I}$, we have $\depth_{\bott B}(b)=\depth_{\prod_{i_\in I} \bott A_i}((a_i)_{i\in I})$.
\end{lem}

\begin{lem}\label{LemTruncUnrolling}
  Let $A,B,C\in\DD_D$, and let $k\geq\depth(A)$. Then $A=BC$ if and
  only if $[\tildeop A]_k=[\tildeop B]_k[\tildeop C]_k$.
\end{lem}
\begin{proof}
  Remember that we already know that $A=BC\iff \tildeop A=\tildeop
  B\tildeop C$. Now, one direction is trivial: if $A=BC$, then $\tildeop
  A=\tildeop B\tildeop C$ so $[\tildeop A]_k=[\tildeop B]_k[\tildeop
    C]_k$ for every $k$. Now, we assume that $[\tildeop A]_k=[\tildeop
    B]_k[\tildeop C]_k$ for some $k \ge \depth(A)$ and we show that $\tildeop A=\tildeop
  B\tildeop C$.

  Now, we want to create a tree product isomorphism
  $\phi: S_{\tildeop A}\rightarrow S_{\tildeop B\tildeop C}$ for the product $\tildeop A=\tildeop B\tildeop C$. To do so,
  we start from the tree product isomorphism $\psi: S_{[\tildeop A]_k}\rightarrow
  S_{[\tildeop B]_k[\tildeop C]_k}$ for the product $[\tildeop A]_k=[\tildeop
    B]_k[\tildeop C]_k$.

  We can extend $\psi$ to $\phi$ easily. For all
  $(a,d)\in S_A\times\N$ where $d\geq \depth_A(a)$, set
  $\phi(a,d)=((b,d),(c,d))$ where
  $\psi(a,\depth_A(a))=((b,\depth_A(a)),(c,\depth_A(a)))$. This is a
  well-defined function since $\psi(a,\depth_A(a))$ will always exist
  as $k\geq\depth(A)$.
  
  Let's prove that this is a valid tree product isomorphism. First, $\phi$ is bijective. Indeed, suppose that $\psi(a,d)=\psi(a',d')$. Denote $\psi(a,d)=((b,d),(c,d))$ and $\psi(a',d')=((b',d'),(c',d'))$. We directly have $(b,c,d)=(b',c',d')$. This means that $\depth_A(a)=\depth_A(a')$, by definition of $\phi$, because $b$ and $c$ are at the same depth as $a$ (this follows from \Cref{TreeProdIsoDepth}, since $\psi$ is a tree product isomorphism).This means that $\psi(a,\depth_A(a))=\psi(a',\depth_A(a))$, which implies $a=a'$ by bijectivity of $\psi$.
  
  Now, for any $(a,d)\in S_A\times\N$ such that $d\geq\depth_A(a)$, $\phi(a,d)=((b,d),(c,d))$ is a root if and only if $d=0$ and $b$ and $c$ are roots. Because of the definition of $\psi$, $b$ and $c$ are roots if and only if $a$ is a root in $A$, since $\psi$ is a tree product isomorphism.
  
  For the last property we need to check, we write $x\xrightarrow{\bott C} y$ to mean that there is an edge from $x\in S_{\bott C}$ to $y\in S_{\bott C}$ in $\bott C$.
  
  Finally, we show that for all $((b,d),(c,d)),((b',d'),(c',d'))\in S_{\tildeop{B}\tildeop C}$, we have: $\phi^{-1}(((b,d),(c,d)))\xrightarrow{\tildeop A} \phi^{-1}(((b',d'),(c',d')))$ if and only if $(b,d)\xrightarrow{\tildeop B} (b',d')$ and $(c,d)\xrightarrow{\tildeop C} (c',d')$. Indeed, following the definition of $\phi$ from $\psi$, we can write $\phi^{-1}(((b,d),(c,d)))=(a,d)\in S_{\tildeop A}$ and $\phi^{-1}(((b',d'),(c',d')))=(a',d')\in S_{\tildeop A}$.
  
  Since $\psi$ is a tree product isomorphism, there is an edge $(a,d)\xrightarrow{\tildeop A} (a',d')$ if and only if there is an edge $((b,\depth_A(a)), (c,\depth_A(a)))\xrightarrow{[\tildeop B]_k[\tildeop C]_k} ((b',\depth_A(a)), (c',\depth_A(a)))$, that is, if and only if there is an edge $((b,d),(c,d))\xrightarrow{\tildeop B\tildeop C} ((b',d),(c',d))$, which is equivalent to the existence of $(b,d)\xrightarrow{\tildeop B} (b',d')$ and $(c,d)\xrightarrow{\tildeop C} (c',d')$.
  
  This proves that $\tildeop A=\tildeop B\tildeop C$, which in turn
  proves that $A=BC$, and concludes.
\end{proof}

\begin{thm}
  Given $A,B\in\DD_D$, we can find $C\in\DD_D$ such that $A=BC$ or
  prove that it does not exist in polynomial time in the sizes of $A$
  and $B$.
\end{thm}
\begin{proof}
  Given $A,B\in\DD_D$, let $k=\depth(A)$. Then, call $divide([\tildeop
    A]_k, [\tildeop B]_k)$. If this function returns $\bot$, then
  there is no $X\in \FF_T$ such that $[\tildeop A]_k=[\tildeop B]_kX$,
  which shows that there is no $C\in\DD_D$ such that $A=BC$ by
  \Cref{LemTruncUnrolling}.

  Otherwise, if this function returns some $X\in\FF_T$, then we have
  $[\tildeop A]_k=[\tildeop B]_kX$ with $\depth(X)=k$. Now, remark
  that if there is some $C\in\DD_D$ such that $A=BC$, we have
  $\depth(C)\leq k$ and thus $[\tildeop A]_k=[\tildeop B]_k[\tildeop
    C]_k$, so by \Cref{FiniteTreeDiv}, we have $X=[\tildeop C]_k$.
  Therefore, if the function returns an $X\in\FF_T$, either $X$ is of
  the form $[\tildeop C]_k$ for some $C\in \DD_D$, and then we recover
  $C$ such that $A=BC$ from the reverse direction of
  \Cref{LemTruncUnrolling}, or $X$ is not of that form, and by
  \Cref{FiniteTreeDiv}, there is no $C\in\DD_D$ such that $A=BC$.
  
  The $divide$ algorithm is indeed in polynomial time since a call to $divide(\bott T, \bott U)$ ends up making at most one call to $divide(\bott V,\bott W)$ for $\bott V$ some subtree of $\bott T$ and $\bott W$ some subtree of $\bott U$. Since every operation in a call to $divide$ is in polynomial time, this concludes.
\end{proof}

\section{Unicity of $k$-th roots} \label{kroot}

Using \Cref{Cancellative}, we can prove a simple result above polynomials,
which in particular states that a polynomial with a coefficient of
degree $1$ which is cancellative is injective.

\begin{prop} \label{Polynomials}
  Let $P=\sum_{i=0}a_iX^i\in\DD[X]$ and $A,B\in\DD$ such that
  $P(A)=P(B)$. Then, we have $A=B$ if $a_1\in\DD^*$ or if for some
  $i>1$, $a_i\in\DD^*$ and $A\in\DD^*$.
\end{prop}
\begin{proof}
  Write $P(X)=\sum_{i=0}^da_iX^i$. We can assume that $a_0=0$, and we
  still have $P(A)=P(B)$. Let $D=\sum_{i=1}^da_i\sum_{j=0}^{i-1}
  A^{i-1-j}B^{j}$. Then:
  \begin{eqnarray*}
    AD &=& \sum_{i=1}^da_i\sum_{j=0}^{i-1} A^{i-j}B^j \\
    &=& \sum_{i=1}^da_i\left(A^i + \sum_{j=1}^{i-1} A^{i-j}B^j\right) \\
    &=& P(A) + \sum_{i=1}^da_i\sum_{j=1}^{i-1} A^{i-j}B^j \\
    &=& P(B) + \sum_{i=1}^da_i\sum_{j=1}^{i-1} A^{i-j}B^j \\
    &=& \sum_{i=1}^da_i\left(B^i + \sum_{j=1}^{i-1} A^{i-j}B^j\right) \\
    &=& \sum_{i=1}^da_i\sum_{j=1}^{i} A^{i-j}B^j \\
    &=& \sum_{i=1}^da_i\sum_{j=0}^{i-1} A^{i-1-j}B^{j+1} \\
    &=& BD.
  \end{eqnarray*}

  In the case where $a_1$ has a fixpoint, remark that the term for
  $i=1$ in $D=\sum_{i=1}^da_i\sum_{j=0}^{i-1} A^{i-1-j}B^{j}$ is simply
  $a_1$, and so, $D$ has a fixpoint. Otherwise, in the case where
  there is $i>1$ such that $a_i$ with a fixpoint, and $A$ has a
  fixpoint, the term in the sum for that $i$ is: $a_i\sum_{j=0}^{i-1}
  A^{i-1-j}B^{j}$, in which we find the term $a_iA^{i-1}$, which has a
  fixpoint, so $D\in\DD^*$.

  Since $D\in\DD^*$, $AD=BD$ implies $A=B$.
\end{proof}

A general characterisation of injective polynomials would be very
interesting. It seems unlikely that the condition $a_1\in\DD^*$ is necessary
since that would mean that if $a_1\notin\DD^*$ then, even if every
other coefficient is in $\DD^*$, one could find $A\neq B$ such that
$P(A)=P(B)$.

In the rest of this section, we show that for any $k \ge 1$, the polynomial $P(X) = X^k$ is injective.

\begin{thm} \label{roots}
  For all $k \ge 1$ and $A,B\in\DD$, if $A^k=B^k$, then $A=B$.
\end{thm}

Our first step is to prove the injectivity of the mapping $\bott X\mapsto \bott X^k$ on $\FF$. Given a forest $\bott F \in \FF$, let $\mathcal{R}( \bott F ) \in \FF_T$ be the tree obtained by joining all the trees of $\bott F$ to a new common root. More formally, if $\Ff( \bott F )$ is the multiset of trees of $\bott F$, then $\Dd( \mathcal{R}( \bott F ) ) = \Ff$.

\begin{lem}\label{LemTreeIsForest}
    For any forest $\bott F \in \FF$ and any $k \ge 1$, we have $\mathcal{R}^k( \bott F ) = \mathcal{R}( \bott F )$.
\end{lem}
\begin{proof}
  By \Cref{LevelByLevelProduct}, we have $\Dd( \mathcal{R}^k( \bott F ) ) = \Ff^k( \bott F)$. Now, it is clear that $\Ff^k( \bott F ) = \Ff( \bott F^k )$. This concludes.
\end{proof}

\begin{lem}
The mapping $\bott X\mapsto \bott X^k$ is injective on $\FF$.
\end{lem}

\begin{proof}
  We first prove that the mapping $\bott X\mapsto \bott X^k$ is injective on $\FF_T$. Let $\bott T_1, \bott T_2 \in \FF_T$ with $\bott T_1 < \bott T_2$. Then by induction on $k$, \Cref{GeneralProdOrder} shows that $\bott T_1^k < \bott T_2^k$.
  
  We now prove injectivity on $\FF$. Let $\bott A,\bott B\in\FF$, such that $\bott A^k=\bott B^k$.
  By \Cref{LemTreeIsForest}, we have $\mathcal{R}^k(\bott A) = \mathcal{R}^k(\bott B)$. By injectivity on $\FF_T$, we obtain $\mathcal{R}(\bott A) = \mathcal{R}(\bott B)$, which implies $\bott A = \bott B$.
\end{proof}

Our second step is to prove the result for bijective FDSs.

\begin{lem}\label{lensRecovery}
  Let $A,B\in\DD$. If $A^k=B^k$, then $[A]_0=[B]_0$.
\end{lem}
\begin{proof}
  Remark that $A^k=B^k$ implies $[A]_0^k=[B]_0^k$. All that's left to
  show is that if $A,B\in\DD_P$ and $A^k=B^k$, then $A=B$.

  Take $D \in\DD_P$, and write $D =\sum_{i} \lambda^A_iC_i$. Assume
  there exists $B=\sum_i \lambda^B_i C_i$ such that $B^k = D$. 
  For all $i \in \N$, let $F_i = \{ L = (l_j)_{j \in \iitv{1,k}}:\bigvee_j l_j=i\}$ denote the possible ways a product of $k$ cycles $C_{l_1} \times \dots \times C_{l_k}$ is equal to some scalar multiple of $C_i$. 
  
  For any sequence $L = (l_j)$, we abuse notation and identify $L$ with the multiset of its entries; we can then use the notation $\delta_L$. By \Cref{LemDelta}, we obtain for all $i \in \N$
  \[
    \sum_{L \in F_i} \delta_L \prod_{j=1}^k \lambda^B_{l_j} = \lambda^A_i.
  \]
  This is a set of triangular positive polynomial equations (as the equation for $i$ only involves $\lambda^B_1, \dots, \lambda^B_i$), thus it has at most one solution. Therefore, if $B$ exists, it is unique.
\end{proof}

Our third and final step proves the theorem. 

\begin{lem}\label{SmallestLengthExtractPoly}
  Let $P\in\N[X]$ be a polynomial with coefficients in $\N$, and let
  $A\in\DD$. Then, for any $\ell\in\N$, we have
  $\supp_{\leq\ell}(P(A))=P(\supp_{\leq\ell}(A))$.
\end{lem}
\begin{proof}
  Write $P=\sum_{i=1}^d a_iX^i$. If $A=\sum_{j=1}^n A_j$ where each
  $A_j$ is connected, then the products that appear in $P(A)$ are the
  $a_i\prod_{k=1}^i A_{\beta_k}$ for each $i\in\iitv{1,n}$ and
  $\beta = (\beta_k)_{k \in \iitv{1,i}} \in\iitv{1,n}^i$. Remark that for such a product
  $a_i\prod_{k=1}^i A_{\beta_k}$ to have a cycle length $\leq\ell$,
  every $A_{\beta_k}$ must have cycle length $\leq\ell$. This
  concludes.
\end{proof}

\begin{lem}\label{LemFF_TToDD}
   The mapping $A\mapsto
   A^k$ is injective on $\DD$.
\end{lem}
\begin{proof}
  Given $A^k$, we find $[A]_0$ and thus we know the lengths of the
  cycles of $A$ by \Cref{lensRecovery}; denote them
  $\ell_1<\dots<\ell_n$. We show by induction on $i\in \iitv{0,n}$ that we can
  recover $\supp_{\leq \ell_i}(A)$ from $A^k$ (with an implicit
  $\ell_0=0$, such that $\supp_{\leq \ell_0}(A)=0$, to make for a
  trivial base case and avoid repetition).

  Take some $i \in \iitv{1,n-1}$ such that the induction hypothesis
  stands for $i$. We show that it also stands for $i+1$. By
  \Cref{SmallestLengthExtractPoly},
  $\supp_{\leq\ell_{i+1}}(A^k)=(\supp_{\leq\ell_{i+1}}(A))^k$. By the
  lemma's hypothesis, we recover $\tildeop{\supp_{\leq\ell_{i+1}}(A)}$
  from $\tildeop{(\supp_{\leq\ell_{i+1}}(A))^k}$. Now, since we have
  $\tildeop{\supp_{\leq\ell_i}(A)}$ from the induction hypothesis, we
  recover
  $$\tildeop{\supp_{\ell_{i+1}}(A)}=\tildeop{\supp_{\leq\ell_{i+1}}(A)}\setminus\tildeop{\supp_{\leq\ell_{i}}(A)}.$$

  It is straightfoward to reconstruct $\supp_{\ell_{i+1}}(A)$ from
  $\tildeop{\supp_{\ell_{i+1}}(A)}$ since we know there every tree in
  $\tildeop{\supp_{\ell_{i+1}}(A)}$ comes from a connected component
  of cycle length $\ell_{i+1}$. And thus, we recover
  $\supp_{\leq\ell_{i+1}}(A)=\supp_{\ell_{i+1}}(A)+\supp_{{\leq\ell_i}(A)}$,
  which concludes the induction.
\end{proof}

\section{A family of monoids with unique factorisation}\label{sec:LD_K}

The $C_2^2=2C_2$ identity shows that factorisation into irreducible
FDSs is not unique on $\DD$. Moreover, it is shown in \cite{Couturier}
that factorisation is also not necessarily unique on $\DD_D$, for
example with the identity presented in \Cref{CexCouturier}. We
can however exhibit an example of an interesting class of trees in
which every element has a unique factorisation in irreducible FDSs.
Although our example might not be useful in practice, it is interesting as a
generalisation of the simpler result that shows that factorisation is
unique on the multiplicative monoid generated by products of paths
(which is called $LD_1$ with the notations below).

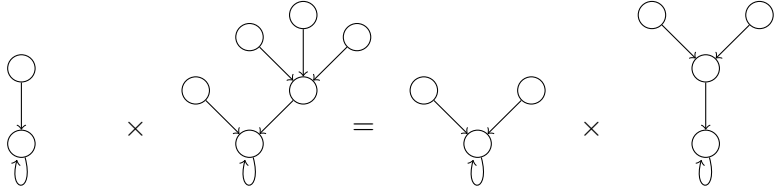
\begin{figure}
  \begin{center}
    \begin{tikzpicture}[c/.style={shape=circle, draw=black}]      
      \node[c] (C1) {}; \node[c] (C2) [above of=C1] {}; \draw[->] (C2)
      -- (C1); \draw[->] (C1) edge[loop below] (C1);

      \tikzset{xshift=1.5cm}
      \node (times1) at (0, 0.2) {$\times$};
      \tikzset{xshift=1.5cm}

      \node[c] (B1) {};
      \node[c] (B2) [above left of=B1] {};
      \node[c] (B3) [above right of=B1] {};
      \node[c] (B4) [above left of=B3] {};
      \node[c] (B5) [above right of=B3] {};
      \node[c] (B6) [above of=B3] {};
      \draw[->] (B2) -- (B1);
      \draw[->] (B3) -- (B1);
      \draw[->] (B4) -- (B3);
      \draw[->] (B5) -- (B3);
      \draw[->] (B6) -- (B3);
      \draw[->] (B1) edge[loop below] (B1);

      \tikzset{xshift=1.5cm}
      \node (eq1) at (0, 0.2) {$=$};
      \tikzset{xshift=1.5cm}

      \node[c] (A1) {};
      \node[c] (A2) [above left of=A1] {};
      \node[c] (A3) [above right of=A1] {};
      \draw[->] (A2) -- (A1);
      \draw[->] (A3) -- (A1);
      \draw[->] (A1) edge[loop below] (A1);
      
      \tikzset{xshift=1.5cm}
      \node (times2) at (0, 0.2) {$\times$};
      \tikzset{xshift=1.5cm}

      \node[c] (D1) {};
      \node[c] (D2) [above of=D1] {};
      \node[c] (D3) [above left of=D2] {};
      \node[c] (D4) [above right of=D2] {};
      \draw[->] (D2) -- (D1);
      \draw[->] (D3) -- (D2);
      \draw[->] (D4) -- (D2);
      \draw[->] (D1) edge[loop below] (D1);
    \end{tikzpicture}
  \end{center}
  
  \caption{A dendron that admits two different factorisations in irreducible factors.}
  \label{CexCouturier}
\end{figure}

\begin{defn}
    A \textit{rhizome} is a path from a leaf to the fixpoint in a dendron.
  The length of a rhizome is its number of transitions, that is its
  number of non-fixpoint states.

%  The \textit{depth} of a tree $A$ the length of its longest rhizome.
\end{defn}

According to our terminology, the depth of a dendron is the length of its longest rhizome.

\begin{defn}
  An FDS $A\in\DD$ is a \textit{linear dendron} if it is a dendron
  such that only its fixpoint may have more than one predecessor. A
  linear dendron has $K$ rhizomes if its fixpoint has $K$ non-fixpoint
  predecessors.

  A \textit{star} $S_n$ is a linear dendron of depth $1$ and $n$ states, while a \textit{path} $P_n$ is a linear dendron with only one rhizome and $n+1$ states.
\end{defn}

We are now in position to show that most linear dendrons are irreducible. We remark that the semigroup of stars is isomorphic to that of the positive integers: $S_{ab} = S_a \times S_b$. Therefore, composite stars have a unique factorisation in $\DD$. 

\begin{prop}\label{irreducible_dendrons}
The only reducible linear dendrons are the stars with a composite number of states.
\end{prop}

\begin{proof}
The case of stars is straightforward. Let $T$ be a linear dendron of depth $k > 1$. Then any rhizome of maximum length of $T$ contains a state with exactly one predecessor: the state at depth $1$ of the rhizome.

Suppose $T$ is reducible towards a contradiction, say $T = A \times B$. The depth of either $A$ or $B$ is at least $k$, say $P_k$ is a subdendron of $A$. Moreover, $P_1$ is a subdendron of $B$. Thus, $P_k\times P_1$ is a subdendron of $A\times B=T$. It's easy to see that
  $P_k\times P_1$ contains a path of depth $k$ states with more than
  one predecessor each (except the leaf at the end). This is a rhizome
  of maximal length in $T$ in which no state has exactly one
  predecessor. This concludes.
\end{proof}

\begin{defn}
  For all $K\in\N$, we define $LD_K$ the multiplicative monoid
  generated by linear dendrons with $K$ rhizomes.
\end{defn}

Based on \Cref{irreducible_dendrons}, if $P \in LD_K$ has a unique factorisation in $LD_K$, then it has a unique factorisation in $\DD$. Thus, we focus on factorisation in $LD_K$.

Let $P \in LD_K$ be factorised as $P = F_1 \times \dots \times F_N$ where $F_j$ is a linear dendron for each $1 \le j \le N$. Each state $s \in S_P$ can be expressed as $s = (s_1, \dots, s_N)$ where $s_j \in S_{F_j}$ for all $j$. Some of those $s_j$'s could be fixed points; let $I(s) = \{ j : F_j(s_j) = s_j \}$. Then the number of predecessors of $s$ is either $0$ if any $s_j$ is a leaf, or equal to $(K+1)^{|I(s)|}$ otherwise. This suggests the following notation.

\begin{defn}
  Let $P \in LD_K$ and $i\in\N$. A state $s$ of $P$ is \textit{$i$-fixed} if it
  has $(K+1)^i$ predecessors.
\end{defn}

\begin{lem}
  Any $i$-fixed state has a unique $i$-fixed predecessor; all other predecessors are either leaves or $j$-fixed for some $j < i$.
\end{lem}

\begin{proof}
Let $s = (s_1, \dots, s_N)$ be $i$-fixed and without loss let $I(s) = \iitv{1,i}$. Remark that $s_1, \dots, s_i$ are fixed points, while $s_{i+1}, \dots, s_N$ have a unique predecessor each, say $t_{i+1}, \dots, t_N$ respectively. Then any predecessor of $s$ is of the form $u = (u_1, \dots, u_i, t_{i+1}, \dots, t_N)$  where $u_l$ is a predecessor of $s_l$ for all $1 \le l \le i$. Therefore $u$ is at most $i$-fixed, with equality if and only if $u = (s_1, \dots, s_i, t_{i+1}, \dots, t_N)$.
\end{proof}

Now that we have all the necessary definitions, we can introduce the
following lemma, which enables a partial recovery of some factors from
a product of linear dendrons. This is the core lemma, and it is from
it that we can finally recover every factor.

\begin{lem}[Linear extraction lemma]\label{LemExtract}
Let $P = F_1 \times \dots \times F_N \in LD_K$ and let $s$ be a depth $1$, codepth $\ell$, $i$-fixed state of $P$. Consider the tree anchored on $s$ in $P$ and remove the unique $i$-fixed predecessor of $s$ and all its antecedents. Then the obtained dendron is
  $E_s = [ \prod_{j \in I(s)} F_j ]_\ell$.
\end{lem}

\begin{proof}
Without loss, let $I(s) = \iitv{1,i}$. Denote $s = (s_1, \dots, s_n)$ and for all $i+1 \le j \le N$ and $d \in \N$ let $t^d_j$ be the unique state of $F_j$ satisfying $F_j^d( t^d_j ) = s_j$. All the states in the dendron $E_s$ are either $s$ or of the form $u = (u_1, \dots, u_i, t^d_{i+1}, \dots, t^d_N)$, where $d$ is the depth of $u$ in $E_s$ and $(u_1, \dots, u_i) \ne (s_1, \dots, s_i)$. By removing the coordinates $i+1, \dots, N$ from each state, we see that $E_s$ is a sub-FDS of $F_1 \times \dots \times F_i$.

  All that is left is to show that we do indeed get the truncature at
  depth $\ell$. Remark that the codepth $\ell$ of $s$ is the length of
  the smallest path among the rhizomes anchored at $s_{i+1}, \dots, s_N$
  in their respective factors, minus $1$. Thus, the sub-FDS of
  $F_1\times \dots\times F_i$ we obtain is indeed truncated at depth $\ell$.
\end{proof}

Let $P = F_1 \times \dots \times F_N$ where all the factors have depth $k+1$. Let $\mathfrak{A}=\{[F_i]_k:i\in\iitv{1,N}\}$ be the collection of truncated factors and for each $B \in \mathfrak{A}$, denote its multiplicity $n_B = |\{ i \in\iitv{1,N} : [F_i]_k = B \}|$. We denote $D_i$ the set of $i$-fixed depth 1 states of codepth $k$ of $P$.

\begin{lem} \label{extraction_powers}
  For all $B \in \mathfrak{A}$, there exists $s \in D_i$ with $E_s = B^i$ if and only if $i \le n_B$.
\end{lem}

\begin{proof}
  Without loss, let $B\in\mathfrak{A}$ such that $B = F_1 = \dots = F_{n_B}$.
  Let $i \le n_B$ and consider a state $s = (s_1, \dots, s_N)$ of $P$ where $s_1, \dots, s_i$ are fixed points of $B$, while for every $i+1 \le j \le N$, $s_i$ is a depth $1$ state on
  a path of depth $k+1$. Then $s \in D_i$,
  and the extraction lemma extracts $B^n$ from $s$. Conversely, if $E_s = B^j$, then $B^j$ divides $[P]_k$ and hence $j \le n_B$.
\end{proof}

We now show
that factorisation is unique on products of linear dendrons which share the same
depth.

\begin{lem}
  A product of elements of $LD_K$ which have the same depth $k$ is
  uniquely factorisable.
\end{lem}
\begin{proof}
    We do this by induction on the depth $k$. For $k=0$, this lemma is
    obvious (the factorisation is $C_1$). Take some $k$ such that the
    lemma stands for depth $k$. We show that the lemma is also true
    for depth $k+1$. The proof is in four steps. First, we identify
    the number of factors, then we recover the set of their depth $k$
    truncatures, then we recover the multiset of these truncatures and
    finally, we recover the full, untruncated factors. Take $P$ a
    product of elements of $LT_K$.

  \head{Number of factors} We recover $N$ the number of factors of $P$
  by remarking that its fixpoint is $N$-fixed: thus by counting its
  number of predecessors, we can recover $N$ from $P$ and write
  $P=F_1\times \dots \times F_N$. 
%   In the following, we use a product
%   isomorphism $\phi: S_P\mapsto S_{F_1}\times \dots \times S_{F_N}$.

  \head{Set of truncatures} 
  According to \Cref{extraction_powers}, by applying the extraction lemma to all the elements of $D_1$, we recover all the factors $B \in \mathfrak{A}$.

\head{Multiset of truncatures} By \Cref{extraction_powers}, for all $B \in \mathfrak{A}$, $n_B = \max\{ i : \exists s \in D_i E_s = B^i \}$. As such, applying the extraction lemma on $D_i$ for $1 \le i \le N$ then yields $n_B$ for all $B \in \mathfrak{A}$.

  \head{Untruncated factors} As of now, we have all the
  factors and their multiplicity, but they are truncated at depth $k$. To fully reconstruct the linear dendron $F_i$ of depth $k+1$ from $[F_i]_k$, all we need is the number $f_i$ of paths of depth $k+1$ in $F_i$. We now show how to determine this number.
%   In other words, we almost have the factors, but we cannot
%   distinguish paths of depth $k$ from paths of depth $k+1$. Let's find
%   that information.

  Fix $B\in \mathfrak{A}$. Let's denote $f_1, \dots, f_{n_B}$ the
  number of paths of depth $k+1$ of and let $G_1, \dots, G_{n_B}$ be
  the elements of $\phi(B)$. For any $n\in\iitv{0,n_B}$, let's count
  in $P$ the number of states of $D_{N-n}$ from which the extraction
  lemma extracts $[P]_k/B^n$. 
  Each of these states
  corresponds to an $n$-uple of depth 1 states of $G_1, \dots,
  G_{n_B}$ (each in a distinct factor) on which a path of depth $k+1$
  is anchored. As such, there are
  $p_n:=\sum_{\substack{I\subs\iitv{1,n_B}\\|I|=n}}\prod_{i\in I} f_i$
  of them (given the set of factors of $G_1, \dots, G_{n_B}$ of index in $I$, the number of depth $1$ states of codepth $k+1$ is $\prod_{i\in I} f_i$). Finding that number for all $n\in\iitv{0,n_B}$ makes it
  possible to express the $f_1, \dots, f_{n_B}$ as the roots of a
  polynomial of degree $n_B$ and thus, allows one to find them. Here
  is how we proceed. Write $R(X) = \sum_{m=0}^{n_B} (-1)^mp_mX^{n_B-m}$. By
  Vieta's relations, we know that the $n_B$ roots of $R$ are $f_1,
  \dots, f_{n_B}$.
\end{proof}

Now, we show that we can always get to this case:

\begin{thm}
  Factorisation is unique on $LD_K$.
\end{thm}
\begin{proof}
  Let $P = F_1 \times \dots \times F_N \in LD_K$ have depth $k+1$. Let $I = \{ 1 \le i \le N : \depth(F_i) \le k \}$ be the set of indices of factors with no paths of depth $k+1$. Now, let $S$ be the set of
  depth 1 states belonging to a rhizome in $P$ of depth $k+1$. For all $s \in S$, since $s$ has depth $1$ and codepth $k$, $s_i$ is a fixpoint for all $i \in I$ and hence $s$ is at least $|I|$-fixed. Conversely, if $s \in S$ such that $s_j$ has codepth $k$ for all $j \notin I$, then $s$ is $|I|$-fixed.
%   in which
%   each state has a number of predecessors different than $1$.
  %We see that $x$ is an $i$-fixed state, for $i$ such that $|\pred(x)|=(K+1)^i$.
  Using the extraction lemma on such a state $s$, we recover
  $[\prod_{i\in I} F_i]_k = \prod_{i\in I} F_i$.

  Let's divide $S$ by $\prod_{i\in I} F_i$. The result is unique by
  \Cref{UniqTreeDiv}. So, we get $\prod_{j\notin I} F_j$ the product
  of the factors of depth $k+1$, and $\prod_{i\in I} F_i$ the product
  of the factors of depth at most $k$. An induction on the second
  subproduct means that we can extract all the subproducts of shared
  depth, and apply the previous lemma on each of them.
\end{proof}

\section{Conclusion}\label{sec:Directions}

In this article, we have obtained results which may lead to a deeper understanding of the structure of the semiring of FDSs $\DD$. In particular, we have characterised the cancellative elements of $\DD$, shown how to perform division of dendrons in polynomial time, proved that $k$-th roots are unique, and we have exhibited a family of monoids with unique factorisation. While this sheds some light on the structure of $\DD$, there are still many questions.

An interesting direction is the complexity of division on general FDSs, or on cycles. Contrary to the situation on trees, this algorithmic problem may not be in $\mathsf{P}$. On the other hand, it is clearly in $\mathsf{NP}$. The question of knowing whether it is $\mathsf{NP}$-complete is still open, as a reduction (if it exists) does not seem obvious at all.

Another important direction to better understand the structure of $\DD$ is the study of primality, defined as follows: $A\in \DD$ is prime if and only if for every $B,C\in\DD$, $A|BC$ implies $A|B$ or $A|C$. Most of the work on this has been done in \cite{Couturier}, in which Couturier proves that for an FDS to be prime, it must be a dendron. Still, as of now, no example of a prime FDS is known, and no finite-time algorithm to check primality is known.

One could also be interested in more practical applications of FDS factorisation. Imagine for example a "grey box" (some deterministic mechanism that does not display its internal workings, but displays its state such that two different states can always be recognized) that is observed by a probe that records the evolution of its state, until this state falls into a cycle, at which point the probe launches the process again, and so on. Thus, the probe reconstructs the FDS $S$ governing the evolution of the grey box's state. We are interested in a way to know, with the current partial recovery of $S$, how many more states we need to add at the minimum in order to get a factorisable system. This is useful because suppose that the probabilistic model of exploration shows that there is a $90\%$ chance that the probe has recovered at least $90\%$ of the states of $S$. Then, if we know that, say, in order to get a factorisable recovered system, we need to add at least $30\%$ more states than the ones we already have recovered, we know that with probably at least $90\%$, the grey box is not factorisable, that is, it does not contain two independent mechanisms running in parallel.

% \bibliographystyle{plain}
% \bibliography{bibi}

%%%% Copy-pasting the bbl file here to make arxiv submission easier

\end{document}